\documentclass[%
 reprint,
superscriptaddress,
 amsmath,amssymb,
 aps,
]{revtex4-2}
\usepackage[inline]{enumitem}
\usepackage{xcolor}
\usepackage{csquotes}
\usepackage{mathtools} %
\usepackage{physics}
\usepackage{graphicx}%
\usepackage[colorlinks=true, urlcolor=blue, linkcolor=blue, citecolor=blue, pdftex]{hyperref}
\usepackage{cleveref}
\usepackage{amsthm} %
\usepackage{orcidlink}
\usepackage{tikz}
\usetikzlibrary{calc}
\usepackage{placeins}

\theoremstyle{plain}
\newtheorem{thm}{Theorem}
\crefname{thm}{theorem}{theorems}
\Crefname{thm}{Theorem}{Theorems}

\newtheorem{lem}{Lemma}
\crefname{lem}{lemma}{lemmas}
\Crefname{lem}{Lemma}{Lemmas}

\newcommand{\sq}[2]{\filldraw[fill=gray!30,draw=black] (#1,#2) rectangle ++(1,1);}

\newcommand{\edgeSet}{\mathcal{E}}
\newcommand{\plaquetteSet}{\mathcal{P}}
\newcommand{\vertexSet}{\mathcal{V}}

\newcommand{\pdagger}{\phantom{\dagger}}

\newcommand{\aqa}{$\langle aQa ^L\rangle $ Applied Quantum Algorithms, Universiteit Leiden}
\newcommand{\lorentz}{Instituut-Lorentz, Universiteit Leiden, Niels Bohrweg 2, 2333 CA Leiden, Netherlands}
\newcommand{\princetonphysics}{Department of Physics, Princeton University, Princeton, NJ 08544, USA}
\newcommand{\princetontheory}{Department of Electrical and Computer Engineering, Princeton University, Princeton, NJ 08544, USA}
\newcommand{\ulm}{Institute for Complex Quantum Systems, Ulm University, 89069 Ulm, Germany}
\newcommand{\iqst}{Center for Integrated Quantum Science and Technology (IQST), Ulm-Stuttgart, Germany}

\begin{document}
\preprint{APS/123-QED}

\title{Tailoring Bell inequalities to the qudit toric code and self testing}

\author{Elo\"ic Vall\'ee\,\orcidlink{0009-0005-2513-152X}}
\email{vallee@lorentz.leidenuniv.nl}
\affiliation{\aqa}
\affiliation{\lorentz}
 
\author{Owidiusz Makuta\,\orcidlink{0000-0002-0070-8709}}
\affiliation{\aqa}
\affiliation{\lorentz}

\author{Patrick Emonts\,\orcidlink{0000-0002-7274-4071}}
\affiliation{\aqa}
\affiliation{\lorentz}
\affiliation{\ulm}
\affiliation{\iqst}

\author{Rhine Samajdar\,\orcidlink{0000-0001-5171-7798}}
\affiliation{\princetonphysics}
\affiliation{\princetontheory}

\author{Jordi Tura\,\orcidlink{ 0000-0002-6123-1422}}
\affiliation{\aqa}
\affiliation{\lorentz}

\date{\today}

\begin{abstract}
Bell nonlocality provides a robust scalable route to the efficient certification of quantum states.
Here, we introduce a general framework for constructing Bell inequalities tailored to the $\mathbb{Z}_d$ toric code for odd prime local dimensions.
Selecting a suitable subset of stabilizer operators and mapping them to generalized measurement observables, we compute multipartite Bell expressions whose quantum maxima admit a sum-of-squares decomposition.
We show that these inequalities are maximally violated by all states in the ground-state manifold of the $\mathbb{Z}_d$ toric code, and determine their classical (local) bounds through a combination of combinatorial tiling arguments and explicit optimization. 
As a concrete application, we analyze the case of $d=3$ and demonstrate that the maximal violation self-tests the full qutrit toric-code subspace, up to local isometries and complex conjugation. This constitutes, to our knowledge, the first-ever example of self-testing a qutrit subspace. 
Extending these constructions, we further present schemes to enhance the ratio of classical--quantum bounds and thus improve robustness to experimental imperfections. 
Our results establish a pathway toward device-independent certification of highly entangled topological quantum matter and provide new tools for validating qudit states in error-correcting codes and quantum simulation platforms.
\end{abstract}

\maketitle

\section{Introduction} \label{sec:introduction}

A central challenge in quantum information science lies in \textit{verifying} whether an experimental device faithfully prepares an intended quantum state~\cite{carrasco2021theoretical,huang2025certifying}. 
Traditional characterization techniques in this regard typically depend on detailed knowledge of the internal mechanisms of the device. 
Unfortunately, such approaches quickly become infeasible as systems scale up and physical Hilbert spaces deviate from idealized theoretical models giving rise to errors such as qubit leakage in superconducting platforms~\cite{blais_circuit_2021}.

Bell nonlocality~\cite{brunner_bell_2014} provides a natural framework for addressing this problem. 
Beyond its foundational significance, it underpins the concept of self-testing, a device-independent approach that certifies quantum states and operations without requiring assumptions about their internal implementation~\cite{supic_self-testing_2020}. 
By offering a practical, assumption-free certification tool, self-testing enables the reliable validation of quantum devices, even in complex regimes where device-dependent methods can fail.

Early work on self-testing primarily focused on two-qubit states~\cite{mayers_self_2004,mckague_robust_2012}.
With the rapid growth of the field, methods for self-testing bipartite high-dimensional entangled states have been developed~\cite{kaniewski_maximal_2019,sarkar_self-testing_2021,meyer_robustly_2025}.
More recent developments have extended these methods to encompass graph states~\cite{baccari_scalable_2020}, including approaches applicable to systems of any prime local dimension~\cite{santos_scalable_2023}, as well as to the stabilizer subspaces of qubits~\cite{makuta_self-testing_2021,guo_certification_2025}.
Experimental advances have also kept pace: the two-qubit maximally entangled state has been self-tested in experiments closing all loopholes on different platforms~\cite{hensen_loophole-free_2015,giustina_significant-loophole-free_2015,shalm_strong_2015}. 
Despite this progress, advancing self-testing more generally remains challenging. The algebraic complexity of self-testing constructions increases rapidly with both the number of parties and the local dimension, and no general method currently exists for designing Bell inequalities tailored to arbitrary systems. 

In particular, identifying efficient certification techniques is an especially relevant but largely unexplored direction for two broad classes of systems.
First, in setups such as circuit QED~\cite{blais_circuit_2021} or bosonic codes~\cite{michael_new_2016}, qubits are often encoded within \textit{subspaces} of larger Hilbert spaces. 
From an experimental standpoint, it is therefore crucial to develop self-testing methods capable of certifying such encoded subspaces, including higher-dimensional encodings such as qutrits. 
Second, relatively little attention has been devoted to the certification of \emph{topological} quantum states using nonlocality-based approaches.

One of the most extensively studied models of topological quantum matter is the celebrated toric code~\cite{kitaev_fault-tolerant_2003}. 
This exactly solvable model is known to realize an exotic  topologically ordered phase, characterized by long-range entanglement and fractionalized anyonic excitations~\cite{sachdev_kagome-_1992,wen_quantum_2002}.
Experimentally, such toric code states have been realized in both analog and digital quantum simulators, using neutral atoms~\cite{semeghini_probing_2021} and superconducting qubits~\cite{satzinger_realizing_2021}, respectively. 
Beyond its significance as a model of $\mathbb{Z}_2$ topological order, however, the toric code is also one of the simplest examples of a quantum error-correcting code. 
On a topologically nontrivial manifold with genus $g$, the toric code exhibits a protected ground-state degeneracy of $2^{2g}$, which allows it to encode $2g$ logical qubits.
Simplified variants of this logical encoding, such as the surface code, are among the leading paradigms today for fault-tolerant quantum computation and form the architectural basis for many contemporary quantum processors~\cite{acharya_quantum_2025}.

Notably, the toric code can be generalized beyond its qubit implementation---that is, beyond the case of two-level systems per site---to \emph{qudit} systems. 
In this work, we focus on constructing Bell inequalities for the $\mathbb{Z}_d$ toric code with prime $d \geq 3$, which defines a topologically ordered subspace of the full Hilbert space. 
Specifically, we design Bell inequalities that are \emph{maximally violated} by the $\mathbb{Z}_d$ toric code, thereby enabling a fully device-independent certification of the corresponding quantum states.

As a concrete demonstration, we explicitly analyze the self-testing of the $\mathbb{Z}_3$ toric code, establishing a direct connection to earlier work on the $\mathbb{Z}_2$ case~\cite{baccari_device-independent_2020}. 
A key distinction between our approach and that of prior studies such as Ref.~\onlinecite{santos_scalable_2023} lies in the construction method: while the latter relies on the graph-state representation of the system, our framework builds the Bell inequality from a carefully selected subset of stabilizer operators intrinsic to the toric code. 
To the best of our knowledge, this constitutes the first device-independent self-test of a \emph{qutrit subspace}, highlighting the potential of our formalism for certifying higher-dimensional and non-symmetry-breaking quantum states.
Our results have direct implications for both quantum error correction and quantum simulation, as they provide a principled means to certify the correct preparation of topological states---key resources for these applications.

The remainder of this paper is organized as follows. 
In \Cref{sec:preliminaries}, we review the theoretical background on the toric code, Bell inequalities, self-testing, and Bell nonlocality. 
\Cref{sec:bell_inequality_for_toric_code} introduces our method for constructing a Bell inequality tailored to the $\mathbb{Z}_d$ toric code and outlines the procedure for determining its quantum and classical bounds. 
In \Cref{sec:self-testing_the_toric_code}, we present our self-testing results, including an explicit construction for the case of $d=3$. 
\Cref{sec:multiple_special_sites} discusses connections to previous work and possible extensions that yield Bell inequalities with improved robustness to experimental imperfections. 
Finally, in \Cref{sec:conclusion}, we summarize our findings and outline directions for future research.

\section{Preliminaries} \label{sec:preliminaries}
In this section, we introduce the key definitions and background concepts central to our discussion of the toric code, Bell inequalities, and self-testing. 
These preliminaries establish a common framework for the arguments that follow. 
Given the dense notation, \Cref{app:glossary} summarizes all the main symbols used in this manuscript.

\subsection{The toric code} \label{subsec:toric_code}

The $\mathbb{Z}_2$ toric code~\cite{kitaev_fault-tolerant_2003} is a paradigmatic example of a stabilizer-based topological quantum error-correcting code.
Its central operating principle is to encode information in the degenerate ground-state subspace of a Hamiltonian defined on a lattice of $S$\,$=$\,$1/2$ spins, or qubits, on a torus.
Each edge of the lattice hosts a spin, which implies that the total number of spins $N$ must be even under periodic boundary conditions. 
We label qubit sites, vertices, and plaquettes using a double index $(i,j)$, as illustrated in \Cref{fig:toric_code_setup}. 
The sets of all sites, plaquettes, and vertices are denoted by $\edgeSet$, $\plaquetteSet$, and $\vertexSet$, respectively, with cardinalities $|\edgeSet| = N$ and $|\plaquetteSet| = |\vertexSet| = N/2$. 

The Hamiltonian of the toric code  is defined as
\begin{equation}
    H = - \sum_{(i,j) \in \vertexSet} V^{\pdagger}_{(i,j)} - \sum_{(i,j) \in \plaquetteSet} P^{\pdagger}_{(i,j)},
\end{equation}
where the so-called star and plaquette terms, the stabilizer operators, are given by
\begin{alignat}{2}
\nonumber
V^{\pdagger}_{(i,j)} &\equiv X^{\pdagger}_{(i-1,j)} X^{\pdagger}_{(i,j-1)} X^{\pdagger}_{(i,j+1)} X^{\pdagger}_{(i+1,j)}, \quad &&(i,j) \in \vertexSet,\\
P^{\pdagger}_{(i,j)} &\equiv Z^{\pdagger}_{(i-1,j)} Z^{\pdagger}_{(i,j-1)} Z^{\pdagger}_{(i,j+1)} Z^{\pdagger}_{(i+1,j)}, \quad &&(i,j) \in \plaquetteSet,
\end{alignat}
respectively, and $X$ and $Z$ denote the Pauli operators.
Note that, according to our convention, not every pair $(i,j)$ defines both a star and a plaquette operator.

This construction can be generalized to spins of arbitrary on-site dimension (qudits).
Specifically, we consider a directed lattice (as shown in \Cref{fig:toric_code_setup}) and define the $d$-dimensional Pauli operators, also known as the \emph{shift} and \emph{clock} operators, as
\begin{align}\label{eq:shift_and_clock_operators}
\begin{aligned}
X &\equiv \sum_{k=0}^{d-1} \ketbra{(k+1) \bmod d}{k}, \\
Z &\equiv \sum_{k=0}^{d-1} \omega^k \ketbra{k},
\end{aligned}
\end{align}
where $\omega \equiv \exp(2\pi i / d)$ is the $d$th root of unity. 
Here, $X$ cyclically shifts computational basis states, while $Z$ applies a phase proportional to the state index.
The generalized star and plaquette operators for the $\mathbb{Z}_d$ toric code then read
\begin{alignat}{2} 
\nonumber
V^{\pdagger}_{(i,j)} &\equiv X^{\dagger}_{(i-1,j)} X^{\dagger}_{(i,j-1)} X^{\pdagger}_{(i,j+1)} X^{\pdagger}_{(i+1,j)}, \quad &&(i,j) \in \vertexSet,\\
P^{\pdagger}_{(i,j)} &\equiv Z^{\pdagger}_{(i-1,j)} Z^{\dagger}_{(i,j-1)} Z^{\pdagger}_{(i,j+1)} Z^{\dagger}_{(i+1,j)}, \quad &&(i,j) \in \plaquetteSet.
\label{eq:plaquette_and_vertex_operators}
\end{alignat}
In this definition, the application of the adjoint operator depends on the lattice orientation: for star operators $V$, sites living on links that are oriented toward a vertex are acted upon with $X^{\dagger}$, whereas for plaquette operators $P$, sites on anticlockwise-oriented edges are acted upon with $Z^{\dagger}$.

For convenience, we also introduce a set of additional operators
\begin{equation} \label{eq:extra_operators}
\begin{split}
    E^{\pdagger}_{(i,j)}(x) \equiv &\, V_{(i,j-1)}^{1-x} P_{(i+1,j)}^x \\
    = &\, X^{x-1}_{(i-1,j-1)} X^{x-1}_{(i,j-2)} [X^{1-x}Z^x]^{\pdagger}_{(i,j)} \\
    & [X^{1-x}Z^{-x}]^{\pdagger}_{(i+1,j-1)} Z^{x}_{(i+1,j+1)} Z^{-x}_{(i+2,j)},
\end{split}
\end{equation}
where $x \in \{2,\dots,d-1\}$ and $(i,j)$ corresponds to a site on a vertical edge.
A visualization of the operators $V$, $P$, and $E(x)$ is shown in \Cref{fig:toric_code_setup}.

\begin{figure}
    \centering
    \includegraphics[width=\columnwidth]{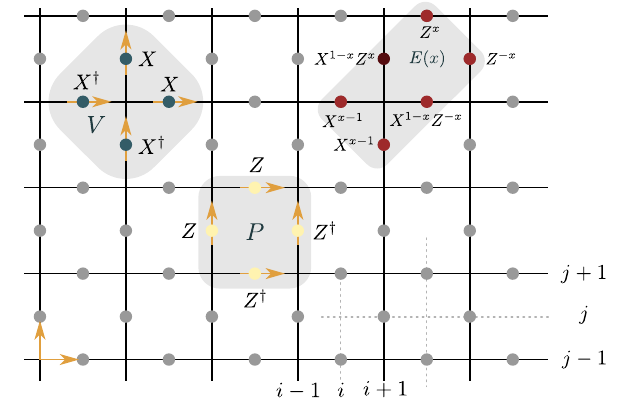}
    \caption{Schematic illustration of the $\mathbb{Z}_d$ toric code on a directed toroidal lattice.
        Orange arrows indicate the directionality of the links.
        The two main types of stabilizer operators, $V$ and $P$, are marked in blue and yellow, respectively.
        An additional operator, $E(x)$, is shown in red and indexed by the \enquote{special site} (dark red); this operator plays an important role in the construction of the Bell inequality.
        The indexing convention is depicted in the bottom right corner with dashed lines.
    }
    \label{fig:toric_code_setup}
\end{figure}

The ground-state manifold of the toric code forms a stabilizer subspace of dimension $d^2$~\cite{gottesman_stabilizer_1997,kitaev_fault-tolerant_2003}.
The set of stabilizing operators is generated by 
\begin{equation} \label{eq:toric_code_generators}
    \left\{V^{\pdagger}_{(i,j)} \,|\, (i,j) \in \vertexSet \right\} \cup  \left\{P^{\pdagger}_{(i,j)} \,|\, (i,j) \in \plaquetteSet \right\}.
\end{equation}
This generating set is not minimal, since $\prod_{(i,j) \in \vertexSet} V_{(i,j)} = \prod_{(i,j) \in \plaquetteSet} P_{(i,j)} = \mathbb{I}$, and this redundancy is precisely what leads to the topologically protected $d^2$-fold ground-state degeneracy.
The additional operators $E_{(i,j)}(x)$ are also stabilizing operators, as they correspond to products of vertex and plaquette operators.

\subsection{Bell inequalities} \label{subsec:bell_inequalities}

To self-test the toric code, one must first construct a suitable Bell inequality. 
In the following, we provide a brief overview of Bell inequalities and the concept of nonlocality.

Nonlocality stipulates that some correlations in nature cannot be described by classical physics~\cite{bell_einstein_1964,brunner_bell_2014}.
It can be revealed through the violation of a Bell inequality.
More precisely, consider a system divided into $N$ subsystems (for instance, a system of $N$ spins), where each subsystem allows for a set of possible measurements (e.g., measurements of spin along the $\hat{x}$ or $\hat{z}$ axes).
For simplicity, we assume that all measurements have the same number of outcomes $d$.
A measurement choice (also called an \emph{input}) is denoted by a vector $\boldsymbol{x}$, where the $i$th entry specifies the measurement setting for the $i$th subsystem, and the measurement outcome is represented by the vector $\boldsymbol{a}$.
We denote by $P(\boldsymbol{a} \vert \boldsymbol{x})$ the probability of obtaining outcome $\boldsymbol{a}$ given input $\boldsymbol{x}$.
The correlations among subsystems are fully characterized by the set of probabilities over all possible inputs and outcomes, $\{ P(\boldsymbol{a} \vert \boldsymbol{x}) \}_{\boldsymbol{a},\boldsymbol{x}}$.

A Bell expression takes the form $\sum_{\boldsymbol{a},\boldsymbol{x}} \alpha_{\boldsymbol{a},\boldsymbol{x}} P(\boldsymbol{a} \vert \boldsymbol{x})$, with $\alpha_{\boldsymbol{a},\boldsymbol{x}} \in \mathbb{R}$.
The local (classical) bounds are denoted by $\beta^{\max}_L$ and $\beta^{\min}_L$, corresponding to the maximal and minimal values achievable using only local correlations.
The quantum bound $\beta^{\max}_Q$ represents the maximal value achievable with quantum correlations.
A convenient metric for comparing Bell expressions is the ratio between quantum and local bounds, defined as
\begin{equation} \label{eq:ratio}
    \Lambda = \frac{\beta^{\max}_Q}{\beta^{\max}_L} \;.
\end{equation}
Bell expressions with larger ratios are generally more effective for witnessing nonlocality, as they offer greater robustness to experimental imperfections in the measurement of correlations.

Bell inequalities can also be expressed in terms of random variables $\{A^{(i)}_{x,k}\}_{i,x,k}$ with $d$ possible outcomes $\{1,\omega, \omega^2, \dots, \omega^{d-1}\}$.
The expected values of these random variables are related to the measurement probabilities via the discrete Fourier transform:
\begin{equation} \label{eq:expval_generalized_observables}
    \left\langle{A^{(1)}_{x_1,k_1}\dots A^{(N)}_{x_N,k_N}} \right\rangle = \sum_{\boldsymbol{a}} \omega^{\boldsymbol{a}\cdot\boldsymbol{k}} P(\boldsymbol{a}|\boldsymbol{x}). 
\end{equation} 
As before, $\omega$ denotes the $d$th root of unity, and $\boldsymbol{k}$ indexes the (Fourier) modes.
Using this formalism, the Bell inequality can be written as
\begin{equation}
\sum_{\boldsymbol{k},\boldsymbol{x}} \tilde{\alpha}_{\boldsymbol{k},\boldsymbol{x}} \left\langle{A^{(1)}_{x_1,k_1}\dots A^{(N)}_{x_N,k_N}} \right\rangle \leq \beta,
\end{equation}
where $\tilde{\alpha}_{\boldsymbol{k},\boldsymbol{x}} = \frac{1}{d}\sum_{\boldsymbol{a}} \omega^{-\boldsymbol{a} \cdot \boldsymbol{k}} \alpha_{\boldsymbol{a},\boldsymbol{x}}$.

As a well-known example, the Clauser-Horne-Shimony-Holt (CHSH) inequality~\cite{clauser_proposed_1969} can be expressed in this language as
\begin{equation}
\expval*{A^{(1)}_{0,1} A^{(2)}_{0,1}} + \expval*{A^{(1)}_{0,1} A^{(2)}_{1,1}} + \expval*{A^{(1)}_{1,1} A^{(2)}_{0,1}} - \expval*{A^{(1)}_{1,1} A^{(2)}_{1,1}} \leq \beta,
\end{equation}
with $\beta_L^{\min} = -2$, $\beta^{\max}_L = 2$, and $\beta^{\max}_Q = 2 \sqrt{2}$.

\subsection{Self-testing} \label{subsec:self-testing}

Beyond distinguishing between local and nonlocal correlations, Bell inequalities can also be used to \emph{self-test} quantum states. 
Self-testing is a method for inferring the underlying quantum state (and measurements) from observed statistics under minimal assumptions~\cite{mayers_self_2004,supic_self-testing_2020}. 
In particular, if a quantum system achieves the maximal violation of a Bell inequality, one can conclude that the system’s state is essentially equivalent to a specific ``ideal'' target state.

Formally, let $\ket{\psi_{\text{ideal}}} \in \mathcal{H}_{\text{ideal}}$ denote the state achieving the maximal violation of a given Bell inequality. 
Consider also a state $\ket{\psi_{\text{max}}} \in \mathcal{H}$ that achieves the same maximal violation. 
If for any such $\ket{\psi_{\text{max}}}$ there exists a local unitary operator $U = U_1 \otimes \dots \otimes U_N$ such that $\mathcal{H} = \mathcal{H}_{\text{ideal}} \otimes \mathcal{H}'$ and
\begin{equation}
U \ket{\psi_{\text{max}}} = \ket{\psi_{\text{ideal}}} \otimes \ket{\psi_{\text{aux}}},
\end{equation}
where $\ket{\psi_{\text{aux}}} \in \mathcal{H}'$ is an auxiliary state, then the Bell inequality is said to \emph{self-test} the state $\ket{\psi_{\text{ideal}}}$. 

More generally, the concept of self-testing can be extended to the verification of subspaces~\cite{baccari_device-independent_2020,makuta_self-testing_2021}. 
In this case, a maximal violation must imply that the state is equivalent, up to local unitaries, to a superposition of states within the target subspace, i.e.,
\begin{equation}
U \ket{\psi_{\max}} = \sum_k \alpha_k \ket*{\psi_{\text{ideal}}^{k}} \otimes \ket*{\psi_{\text{aux}}^{k}},
\end{equation}
where $\sum_k |\alpha_k|^2 = 1$ and $\{\ket*{\psi_{\text{ideal}}^{k}}\}_k$ are orthogonal vectors spanning the subspace to be self-tested.

In cases where the states involve complex phases, i.e., $\ket{\psi^{k}_{\text{ideal}}} \neq \ket{\psi_{\text{ideal}}^{k\, *}}$, self-testing can only be performed up to complex conjugation. 
This is because the correlations obtained from measuring $\ket{\psi^{k}_{\text{ideal}}}$ are identical to those obtained from its complex conjugate $\ket{\psi_{\text{ideal}}^{k\, *}}$. 
This issue, known as the \emph{complex conjugation problem}~\cite{mckague_generalized_2011}, motivates the common practice of extending the definition of self-testing to include equivalence up to complex conjugation. 
Formally, this means showing that there exists a local unitary such that $U \ket{\psi_{\max}}$ can be written as
\begin{equation}
\sum_k \left(\alpha_{1,k} \ket*{\psi_{\text{ideal}}^{k}} \ket*{\psi_{\text{aux}}^{1,k}}
+ \alpha_{2,k} \ket*{\psi_{\text{ideal}}^{k\, *}} \ket*{\psi_{\text{aux}}^{2,k}}\right),  
\end{equation}
where $\sum_{i,k} |\alpha_{i,k}|^2 = 1$, the set of states $\{\ket*{\psi_{\text{ideal}}^{k}} \}_k \cup \{\ket*{\psi_{\text{ideal}}^{k\, *}}\}_k$ spans the ideal subspace (including complex-conjugated states), and $\ket*{\psi_{\text{aux}}^{i,k}}$ are auxiliary states that do not affect observable correlations.

In this manuscript, we adopt this extended definition of self-testing, meaning that states can be certified up to complex conjugation.
The core idea, however, remains unchanged, namely, that self-testing enables the certification of a system without making assumptions about its internal structure. 
Being inherently device-independent, self-testing provides a robust and versatile framework for validating quantum devices.

\section{Bell inequality for the toric code} \label{sec:bell_inequality_for_toric_code}

In this section, we construct a Bell inequality that is maximally violated by all states within the toric code's ground  subspace. 
To this end, we make use of a generating set of stabilizer operators, $\mathcal{S}$, which provides a natural and compact representation of the code space.
Each stabilizer generator $s \in \mathcal{S}$ is associated with a corresponding function of classical random variables, denoted by $\tilde{s}$. 
The Bell inequality is then expressed as a linear combination of these functions:
\begin{equation}
    \sum_{s \in \mathcal{S}} \alpha_{s} \expval*{\tilde{s}} + \text{c.c.} \leq \beta.
\end{equation}
The addition of the complex conjugate is unnecessary when the expression is already real. 

Before deriving the Bell inequality for the $\mathbb{Z}_d$ toric code, it is instructive to review a few simpler examples.

\paragraph{CHSH inequality:} A typical example of a Bell inequality built from stabilizer operators is the CHSH inequality~\cite{clauser_proposed_1969}.
This inequality is maximally violated by the stabilizer state $(\ket{00} + \ket{11}) / \sqrt{2}$, whose generators are $\mathcal{S} = \{X_{(1)}X_{(2)} , Z_{(1)}Z_{(2)} \}$.
Following the construction in~\cite{baccari_device-independent_2020, makuta_self-testing_2021}, we associate to each Pauli operator a linear combination of random variables: $X_1 \rightarrow  (A^{(1)}_{0,1}+A^{(1)}_{1,1})/\sqrt{2}$, $Z_1 \rightarrow (A^{(1)}_{0,1}-A^{(1)}_{1,1})/\sqrt{2}$, $X_2 \rightarrow A^{(2)}_{0,1}$ and $Z_2 \rightarrow A^{(2)}_{1,1}$.
The resultant Bell expression is given by the sum of the stabilizing operators: 
\begin{alignat}{1}
\nonumber
        & X_1X_2 + Z_1Z_2 \rightarrow \expval*{\tilde{s}_1} + \expval*{\tilde{s}_2} \\
        \nonumber
        & = \frac{\expval*{A^{(1)}_{0,1}A^{(2)}_{0,1}} + \expval*{A^{(1)}_{1,1}A^{(2)}_{0,1}} + \expval*{A^{(1)}_{0,1}A^{(2)}_{1,1}} - \expval*{A^{(1)}_{1,1}A^{(2)}_{1,1}} }{ \sqrt{2} } \\
        &\leq \beta^{\max}.
\end{alignat}
This method recovers the standard CHSH expression up to an overall factor.

\paragraph{$\mathbb{Z}_2$ toric code:} A similar approach to that of the previous example, employing analogous substitutions, has been used to derive a Bell inequality for the $\mathbb{Z}_2$ toric code~\cite{baccari_device-independent_2020}.
Again, the Bell expression is obtained by summing over the stabilizing operators:
\begin{equation} \label{eq:bell_inequality_Z2_toric_code}
    \sum_{(i,j) \in \vertexSet} \expval*{\tilde{V}_{(i,j)}} + \sum_{(i,j) \in \plaquetteSet} \expval*{\tilde{P}_{(i,j)}} + \text{c.c.} \leq \beta^{\max}.
\end{equation}
The quantum bound is given by $\beta_Q^{\max} = 2(N-4) + 4\sqrt{2}$, while the classical (local) bound is $\beta_L^{\max} = 2N$. 
Note that the conventional Bell expression does not include the complex conjugate. 
Here, we explicitly define the inequality with the complex conjugate, resulting in a doubling of both the quantum and classical bounds with respect to the original formulation. 
This choice will be important when comparing with the results derived in the subsequent sections.

\subsection{Bell expression}\label{subsec:bell_expression}
In the case of the $\mathbb{Z}_d$ toric code, the construction is a bit more involved than for $\mathbb{Z}_2$, since the number of possible outcomes is now $d$.

Taking inspiration from Refs.~\cite{kaniewski_maximal_2019,santos_scalable_2023}, we consider the set of stabilizing operators
\begin{equation} \label{eq:custom_stabilizing_operators}
\begin{split}
    \mathcal{S} = & \{ V_{(i,j)} : (i,j) \in \vertexSet \} 
    \cup \{ P_{(i,j)} : (i,j) \in \plaquetteSet \} \\
    & \cup \{ E_{(i^*,j^*)}(x) : x = 2,\dots,d-1\},
\end{split}
\end{equation}
where $V_{(i,j)}$, $P_{(i,j)}$, and $E_{(i^*,j^*)}(x)$ are defined in 
\Cref{eq:plaquette_and_vertex_operators,eq:extra_operators}, and $(i^*,j^*)$ denotes a site on a vertical edge. 
We refer to $(i^*,j^*)$ as the \emph{special site}. 
The choice of a vertical edge is arbitrary, as an analogous construction can be performed for a horizontal edge. 
In this work, nonetheless, we fix the convention that the special site lies on a vertical edge. 
For convenience, we define the sets of vertices and plaquettes excluding those adjacent to the special site as 
$\bar{\vertexSet} \equiv \vertexSet \setminus \{(i^*,j^*\pm1)\}$ and $\bar{\plaquetteSet} \equiv \plaquetteSet \setminus \{(i^*\pm1,j^*)\}$, respectively.

To construct the Bell expression, we substitute each stabilizing operator with an appropriate linear combination of random variables. 
Three types of substitutions are required, depending on the site on which the operator acts. 
On the special site $(i^*,j^*)$, we define
\begin{equation} \label{eq:substitutions}
    (X^{1-x}Z^x)^k \rightarrow 
    \bar{A}^{(i^*,j^*)}_{x,k} \equiv 
    \frac{\omega^{-kx(x+1)}}{\sqrt{d}\lambda_{k}} 
    \sum_{y = 0}^{d-1} \omega^{-kxy} A^{(i^*,j^*)}_{y,k},
\end{equation}
where $\lambda_{k}$ is a complex coefficient explained in \Cref{app:complex_coefficients} and first introduced in Ref.~\cite{kaniewski_maximal_2019}. 
At the adjacent site $(i^*+1,j^*-1)$, the substitutions are given by 
\begin{align}\label{eq:substitutions_2}
\begin{aligned}
(X^{1-x}Z^{-x})^k & \rightarrow A^{(i^*+1,j^*-1)}_{x,k} , \quad x \neq 1, \\
    Z^k & \rightarrow A^{(i^*+1,j^*-1)}_{1,k}, \quad x =1. 
\end{aligned}
\end{align}
Finally, for all other sites $(i,j)$, the relevant substitutions are $X^{k} \rightarrow A^{(i,j)}_{0,k}$ and $Z^{k} \rightarrow A^{(i,j)}_{1,k}$. 
As described earlier in \Cref{subsec:bell_inequalities}, $x \in \{0,\dots,d-1\}$ denotes the measurement inputs and $k \in \{1,\dots,d-1\}$ the Fourier modes.
These substitutions are invertible whenever the dimension $d$ is an odd prime number~\cite{kaniewski_maximal_2019}, a property that is crucial for constructing the quantum operators leading to the violation.

Our Bell inequality is then constructed as a linear combination of the stabilizing operators defined in \Cref{eq:custom_stabilizing_operators}, 
where each local operator is replaced by its corresponding generalized observable according to the substitutions in \Cref{eq:substitutions,eq:substitutions_2}. 
The coefficients of the inequality are chosen such that the resulting expression admits a simple sum-of-squares decomposition for the quantum bound (see \Cref{subsec:quantum_bound}).

This procedure yields an $N$-partite Bell inequality. 
The parties corresponding to the sites $(i^*,j^*)$ and $(i^*+1,j^*-1)$ each have $d$ measurement inputs, 
while the remaining $N-2$ parties each have two measurement inputs. 
All parties have $d$ possible outcomes. 
The Bell expression takes the form
\begin{equation} \label{eq:bell_inequality}
{\small
\begin{split}
    \Bigg[ 
    \sum_{(i,j) \in \vertexSet} \expval*{\tilde{V}_{(i,j)}}
    + \sum_{(i,j) \in \plaquetteSet} \expval*{\tilde{P}_{(i,j)}}
    + 2 \sum_{x=2}^{d-1} \expval*{\tilde{E}_{(i^*,j^*)}(x)} 
    \Bigg]
    + \text{c.c.},
\end{split}
}
\end{equation} 
where $\tilde{V}_{(i,j)}$, $\tilde{P}_{(i,j)}$, and $\tilde{E}_{(i^*,j^*)}(x)$ denote the stabilizing operators defined in 
\Cref{eq:plaquette_and_vertex_operators,eq:extra_operators}, with the substitutions from \Cref{eq:substitutions,eq:substitutions_2} applied. 
The precise forms of these terms are detailed in \Cref{app:explicit_bell_expression}.

\subsection{Quantum bound} \label{subsec:quantum_bound}

In the quantum picture, it is convenient to represent the random variables $A_{x,k}$ as quantum operators. 
Following the construction introduced in \Cref{eq:expval_generalized_observables}, we define the \textit{generalized quantum observables}, which may have complex eigenvalues but can be directly related to Hermitian measurement operators via the discrete Fourier transform
\begin{equation}
    A_{x,k} = \sum_{a=0}^{d-1} \omega^{ak} F_{x,a},
\end{equation}
where $\boldsymbol{F}_x \equiv \{F_{x,a}\}_a$ denotes a generalized measurement (a positive operator-valued measure, or POVM).
These observables satisfy the relations $A_{x,0} = \mathbb{I}$, $(A_{x,k})^{\dagger} = A_{x,-k}$, and $A_{x,k+d} = A_{x,k}$ for all $k \in \mathbb{Z}$.

In the special case of $\boldsymbol{F}_x$ corresponding to a projection-valued measurement (PVM), the generalized observables further obey $A_{x,k} = (A_{x,1})^k$. 
The operators then become unitary:
\begin{equation}
    (A_{x,k})^{\dagger}A_{x,k} = (A_{x,1})^{-k}(A_{x,1})^{k} = \mathbb{I}.
\end{equation}
Since no restrictions are imposed on the Hilbert space dimension, Naimark’s dilation theorem~\cite{holevo_statistical_2001} guarantees that one can always construct a PVM and an associated pure quantum state that reproduce the observed correlations. 
Hence, without loss of generality, we may assume that the operators $A_{x,k}$ are unitary.

This assumption greatly simplifies the analysis, allowing the Bell inequality to be expressed in a sum-of-squares (SOS) form, which in turn facilitates the derivation of the quantum bound.

\begin{thm} \label{thm:quantum_bound}
For the Bell expression in \Cref{eq:bell_inequality}, the maximum over all quantum correlations satisfies
\begin{equation}
    \beta^{\max}_Q = 2N + 4d - 8.
\end{equation}
\end{thm}
The proof follows from the following SOS decomposition:
\begin{widetext}
\begin{alignat}{1} \label{eq:sum-of-squares}
\begin{split}
(2N + 4d - 8) 
&- \left[ 
\left( 
\sum_{(i,j) \in \vertexSet} \tilde{V}^{\pdagger}_{(i,j)} 
+ \sum_{(i,j) \in \plaquetteSet} \tilde{P}^{\pdagger}_{(i,j)} 
+ 2 \sum_{x = 2}^{d-1} \tilde{E}^{\pdagger}_{(i^*,j^*)}(x) 
\right) 
+ \text{h.c.} 
\right] \\
&= \left |\mathbb{I} - \tilde{V}^{\pdagger}_{(i^*,j^*-1)}\right|^2 
+ \left|\mathbb{I} - \tilde{V}_{(i^*,j^*+1)}^{\dagger}\right|^2 
+ \left|\mathbb{I} - \tilde{P}_{(i^*-1,j^*)}^{\dagger}\right|^2 
+ \left|\mathbb{I} - \tilde{P}^{\pdagger}_{(i^*+1,j^*)}\right|^2 \\
&+ 2 \sum_{x=2}^{d-1} \left|\mathbb{I} - \tilde{E}^{\pdagger}_{(i^*,j^*)}(x)\right|^2 
+ \sum_{(i,j) \in \bar{\vertexSet}} \left|\mathbb{I} - \tilde{V}^{\pdagger}_{(i,j)}\right|^2 
+ \sum_{(i,j) \in \bar{\plaquetteSet}} \left|\mathbb{I} - \tilde{P}^{\pdagger}_{(i,j)}\right|^2 
\quad \geq 0,
\end{split}
\end{alignat}
\end{widetext}
where $|M|^2 = M^{\dagger}M$. 
The details of the derivation are provided in \Cref{app:quantum_bound}.
Although the substitutions differ slightly, the quantum bound for prime $d \geq 3$ coincides with that of the $d = 2$ case in \Cref{eq:bell_inequality_Z2_toric_code}.

From the SOS form in \Cref{eq:sum-of-squares}, it follows that the bound is saturated precisely when each term in the sum vanishes. 
This condition holds when measuring any state within the ground-state subspace of the toric code using the following generalized observables:
\begin{equation}\label{eq:inverse_substitutions}
\begin{aligned}
    A_{x,k}^{(i^*,j^*)} &= 
        \frac{\lambda_{k}}{\sqrt{d}}
        \sum_{y=0}^{d-1} 
        \omega^{kxy} \omega^{ky(y+1)} (X^{1-y}Z^{y})^k, \\[6pt]
    A^{(i^*+1,j^*-1)}_{x,k} &= 
        \begin{cases}
            (X^{1-x}Z^{-x})^k, & x \neq 1,\\
            Z^k, & x = 1,
        \end{cases} \\[6pt]
    A^{(i,j)}_{0,k} &= X^k, \qquad
    A^{(i,j)}_{1,k} = Z^k.
\end{aligned}
\end{equation}
This construction corresponds to the inverse map of \Cref{eq:substitutions} and is well-defined only when the dimension $d$ is a prime number.

\subsection{Local bound} \label{subsec:local_bound}

The standard approach to computing the local bound of a linear Bell expression involves performing a discrete optimization over all \textit{local deterministic strategies} (LDS).
For correlations arising from LDS, the expected value of the random variables satisfies
\begin{equation}\label{eq:lds_assumption}
    \langle{A^{(1)}_{x_1,k_1} \dots A^{(N)}_{x_N,k_N}}\rangle 
    =\langle{A^{(1)}_{x_1,1}}^{k_1}\rangle  \cdots \langle{A^{(N)}_{x_N,1}}^{k_N}\rangle 
    \,\, \text{(LDS)},
\end{equation}
and the optimization is performed over independent variables 
$\{\expval*{A^{(i)}_{x,1}}\}_{i,x}$, each taking a value in 
$\{1, \omega, \dots, \omega^{d-1}\}$.

For the Bell inequality in \Cref{eq:bell_inequality}, the number of independent variables is $2d + 2(N-2)$ (two parties with $d$ inputs and $N-2$ parties with two inputs). 
Thus, the optimization spans $d^{2N + 2d - 4}$ possible configurations. 
While this is manageable for small $N$ and $d$, the exponential scaling quickly renders it intractable for larger systems. 
To address this, we partition the global optimization into multiple smaller optimizations and verify that the individually optimal solutions are mutually compatible.

Each plaquette and vertex term is optimized separately, except for the terms involving the special site $(i^*,j^*)$, which are optimized jointly. 
This tiling procedure is illustrated in \Cref{fig:toric-code_tiling}(a). 
There are three types of tiles: individual plaquettes far from the special site (gray), single vertices far from the special site (white), and the \textit{special tile} containing all terms associated with the special site (colored region).

\begin{figure}
    \centering
    \includegraphics[width=\columnwidth]{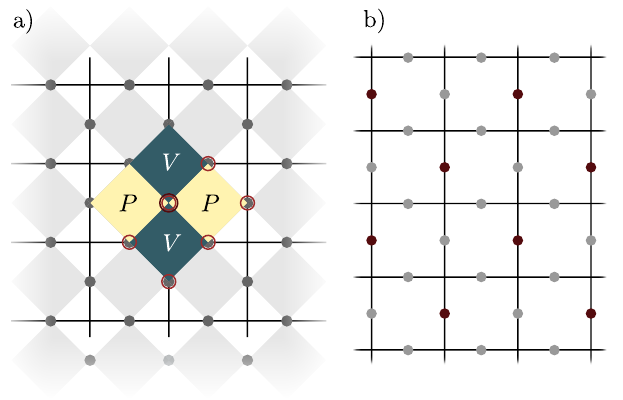}
    \caption{(a) 
    To compute the local bound, the global optimization routine is subdivided into local optimizations. 
    Terms in \Cref{eq:bell_inequality} that reside outside the colored area  are optimized individually, while those inside are jointly optimized.
    The special site is marked by the dark red ring at the center. 
    The $V$ and $P$ operators are represented by colored squares, while the $E$ operator is indicated by bright red rings. 
    (b) 
    Example configuration with multiple special sites, shown in red. 
    All stabilizing operators appearing in the Bell expression in \Cref{eq:bell_inequality_general} include at most one special site.}
    \label{fig:toric-code_tiling}
\end{figure}

The Bell expression restricted to the special tile is given by $\expval{T_{(i_*,j_*)}} + \text{c.c.}$, where 
\begin{alignat}{1} 
\nonumber
    \expval*{T_{(i_*,j_*)}} \equiv &\Bigg[ 
    \expval*{\tilde{V}_{(i^*,j^*-1)}} 
    + \expval*{\tilde{V}_{(i^*,j^*+1)}} 
    + \expval*{\tilde{P}_{(i^*-1,j^*)}} \\
    &+ \expval*{\tilde{P}_{(i^*+1,j^*)}} 
    + 2 \sum_{x=2}^{d-1} \expval*{\tilde{E}_{(i^*,j^*)}(x)} 
    \Bigg] \;.
    \label{eq:special_tile}
\end{alignat}
The whole Bell expression in \Cref{eq:bell_inequality} can then be written as 
\begin{equation}
    \left[ \expval*{T_{(i_*,j_*)}} + \sum_{(i,j) \in \bar{\vertexSet}} \expval*{\tilde{V}_{(i,j)}} +\sum_{(i,j) \in \bar{\plaquetteSet}} \expval*{\tilde{P}_{(i,j)}} \right] + \text{c.c.},
\end{equation}
where $\bar{\plaquetteSet}$ and $\bar{\vertexSet}$ are the sets of plaquettes and vertices far from the special site, as defined in the beginning of \Cref{subsec:bell_expression}.

The optimization of a single vertex (or plaquette) far from the special site involves only four variables, e.g.
\begin{alignat*}{1}
    &\expval*{\tilde{V}_{(i,j)}} + \text{c.c.} \\
    &= 2 \Re\!\left(\!
        \expval*{A^{(i-1,j)}_{0,1}}^{d-1} 
        \expval*{A^{(i,j-1)}_{0,1}}^{d-1}
        \expval*{A^{(i,j+1)}_{0,1}} 
        \expval*{A^{(i+1,j)}_{0,1}}\!
    \right),
\end{alignat*}
where each variable takes a value in $\{1, \omega, \dots, \omega^{d-1}\}$.
Hence, for an individual plaquette or vertex, the maximum value is $2$, while the minimum value is $2 \Re(\omega^{(d-1)/2})$.
Since there are $N-4$ such terms far from $(i^*,j^*)$, we obtain the following result.

\begin{thm} \label{thm:local_bound}
For the Bell expression in \Cref{eq:bell_inequality}, the maximum and minimum over all local correlations satisfy
\begin{equation} \label{eq:local_bound}
\begin{split}
    \beta^{\max}_L & \leq 2(N - 4) + \beta^{\max}_{*}, \\
    \beta^{\min}_L & \geq 2(N - 4)\Re(\omega^{\frac{d-1}{2}}) + \beta^{\min}_{*},
\end{split}
\end{equation}
where $\beta^{\max}_{*}$ ($\beta^{\min}_{*}$) denotes the maximum (minimum) value of the expression $\expval*{T_{(i_*,j_*)}} + \mathrm{c.c.}$, given in \Cref{eq:special_tile}, over all possible local deterministic strategies.

In the special case $d = 3$, these bounds are tight.
\end{thm}

The proof of the first part of the theorem is straightforward. 
We consider the optimal value of the special tile together with the optimal contributions of each individual plaquette and vertex. 
Since optimizing each term separately cannot yield a value smaller than the global maximum, the inequality follows directly.

The second part of the proof involves tessellating the whole grid with elementary tiles (see \Cref{app:local_bound}) 
and constructing optimal strategies for each tile that are mutually compatible. 
In particular, we show that there always exists a strategy achieving the bound in which the boundary variables are all equal to $1$, thereby ensuring compatibility among the tiles. 
Further details are provided in \Cref{app:local_bound}.

We expect that the bounds in \Cref{eq:local_bound} are tight for all prime dimensions $d \geq 2$. 
This expectation is supported by the observation that the number of free variables increases with $d$, making it increasingly likely that a compatible configuration can be found.
Note also that, similar to the quantum bound, the local bound for prime $d \geq 3$ coincides with that of the $d = 2$ case, as given in \Cref{eq:bell_inequality_Z2_toric_code}, despite the use of different substitutions.

Computing the maximal value of the special tile, $\beta_{*}$, is generally nontrivial. 
In the following, we compute this value explicitly for $d=3$.

\subsubsection*{Local bound for $d=3$} \label{subsubsec:local_bound_d=3}

For $d=3$, the optimization in \Cref{eq:special_tile} involves $16$ independent variables, each taking values in $\{1, \omega, \omega^2\}$. 
This yields a search space of $3^{16} \simeq 10^{8}$ configurations.

A brute-force evaluation gives the extremal values
\begin{equation}
    \beta^{\max}_{*} = 12 \cos(\pi/9),
    \quad 
    \beta^{\min}_{*} = -12 \cos(\pi/9).
\end{equation}
The maximum is achieved, for instance, when 
\begin{equation*}
\expval*{A^{(i^*,j^*)}_{0,1}} = \omega^{2}, \quad \expval*{A^{(i^*+1,j^*-1)}_{2,1}} = \omega, 
\end{equation*}
and all other variables are set to $1$. 
Likewise, the minimum occurs when 
\begin{equation*}
\expval*{A^{(i^*,j^*)}_{0,1}} = \expval*{A^{(i^*,j^*)}_{1,1}} = \omega, \quad
\expval*{A^{(i^*,j^*)}_{2,1}} = \omega^2,
\end{equation*}
with all other variables equal to $1$. 
Although other optimal strategies exist, only these two maintain a value of $1$ for all the boundary variables, which is required for compatibility between neighboring tiles.

Consequently, applying \Cref{thm:local_bound}, the local bounds for $d=3$ are
\begin{equation}
    \begin{split}
        \beta^{\max}_L & = 2N - 8 + 12\cos(\pi/9), \\
        \beta^{\min}_L & = -N + 4 - 12\cos(\pi/9),
    \end{split}
\end{equation}
where we used $2\Re(\omega) = -1$.

As an independent verification, we also expressed the difference between the bound and \Cref{eq:special_tile} as a sum of squares. This expression is obtained by relaxing the polynomial optimization problem into a semidefinite program~\cite{gatermann_symmetry_2004,blekherman_semidefinite_2012}.
The code for this construction is available online \footnote{The code supporting this work can be found at the following repository: \url{https://gitlab.com/elo_val/qudit-toric-code-bell-inequality}}. 
The results obtained from this method are consistent with the brute-force computation.

\section{Self-testing the $\mathbb{Z}_3$ toric code} \label{sec:self-testing_the_toric_code}

In this section, we investigate the self-testing of the $\mathbb{Z}_3$ toric code.  
Building on the Bell inequalities introduced in~\Cref{sec:bell_inequality_for_toric_code}, we demonstrate that the ground-state subspace of the toric code can be rigorously certified for local dimension $d=3$.  
Since a description of the qutrit toric code entails states with complex coefficients, it is necessary to adopt an extended definition of self-testing  up to complex conjugation.

\begin{thm}\label{thm:self-testing}
Let $\ket{\psi_{\operatorname{max}}}\in\mathcal{H}$ be a state achieving maximal violation of Eq. \eqref{eq:bell_inequality_general} for $d=3$, where we assume that a partial trace of $\ket{\psi_{\operatorname{max}}}$ over any subset of parties gives a full-rank density matrix.
Then, there exists a unitary $U=\bigotimes_{(i,j)\in\edgeSet}U_{(i,j)}$ such that  $\mathcal{H}=(\mathbb{C}^3)^{\otimes N}\otimes\mathcal{H}'$ and
\begin{equation}
U \ket{\psi_{\max}} = \sum_{k=1}^{4} \left( \alpha_{1,k}\ket{\tau_{k}}\ket{\eta_{1,k}} + \alpha_{2,k}\ket{\tau_{k}^{*}}\ket{\eta_{2,k}} \right),
\end{equation}
where $\alpha_{l,k}\in \mathbb{C}$ are unknown complex coefficients satisfying $\sum_{l=1}^{2}\sum_{k=1}^{4}|\alpha_{l,k}|^{2}=1$.  
The states $\ket{\tau_1}, \dots, \ket{\tau_4} \in \mathbb{C}^{3N}$ form an orthonormal basis for the qutrit toric code, $\ket{\tau_{k}^{*}}$ denotes the complex conjugate of $\ket{\tau_{k}}$, and $\ket{\eta_{l,k}}\in\mathcal{H}'$ are auxiliary states.  
\end{thm}

We now provide a sketch of the proof; full details are given in~\Cref{app:qutrit_self-testing}.  
The central idea is to show that if the state $\ket{\psi_{\max}}$ attains the maximal quantum violation, it must be stabilized by the operators $\tilde{P}_{(i,j)}$, $\tilde{V}_{(i,j)}$, and $\tilde{E}_{(i^*,j^*)}(x)$, each expressible as products of generalized observables.  
Through a sequence of algebraic manipulations, one finds that the operators $\bar{A}_{x,k}^{(i^*,j^*)}$, together with $A_{x,k}^{(i,j)}$ for all $(i,j)\neq(i^*,j^*)$, act as effective Pauli operators when restricted to the support of $\ket{\psi_{\max}}$.  
This correspondence allows us to conclude that $\ket{\psi_{\max}}$ is stabilized by the set of stabilizer operators defining the toric code.  
Consequently, $\ket{\psi_{\max}}$ must lie within the toric-code subspace, up to complex conjugation.

\section{Widening the classical--quantum ratio} \label{sec:multiple_special_sites}

The Bell inequality for the $d$-dimensional toric code introduced in~\Cref{eq:bell_inequality} can be further generalized by incorporating multiple special sites.  
As we show below, this modification leads to an increased classical--quantum ratio $\Lambda$, as defined in~\Cref{eq:ratio}.

Suppose there are $R$ special sites, denoted by $(i_r,j_r)$ with $r \in \{1, \dots, R\}$, and for simplicity, we assume that each lies on a vertical edge.  
The generalized Bell inequality then takes the form
\begin{alignat}{1} 
\nonumber
    &\Bigg[ \sum_{(i,j) \in \vertexSet} \expval*{\tilde{V}_{(i,j)}} 
    + \sum_{(i,j) \in \plaquetteSet} \expval*{\tilde{P}_{(i,j)}} 
    + 2 \sum_{r = 1}^{R} \sum_{x = 2}^{d-1} \expval*{\tilde{E}_{(i_r,j_r)}(x)} \Bigg] \\
    &+ \text{c.c.} \leq \beta.
    \label{eq:bell_inequality_general}
\end{alignat}

To ensure well-defined behavior, each stabilizing operator in the expression must involve no more than one special site.  
In other words, the special sites must be placed sufficiently far apart.  
\Cref{fig:toric-code_tiling}(b) illustrates one possible configuration with multiple special sites, where all stabilizing operators appearing in the Bell expression include at most one of them.

The results of~\Cref{thm:quantum_bound,thm:local_bound,thm:self-testing} naturally extend to the generalized Bell inequality~\eqref{eq:bell_inequality_general}.  
In particular, the quantum bound becomes
\begin{equation}
\beta^{\max}_Q = 2N + (4d - 8)R,
\end{equation}
which follows from a sum-of-squares (SOS) decomposition analogous to the derivation in~\Cref{subsec:quantum_bound} (see also~\Cref{app:quantum_bound} for details).  

For the local bound, we obtain
\begin{alignat}{1}
\beta^{\max}_L &\leq 2(N - 4R) + R\beta_{*}^{\max}, \\\qquad 
\beta^{\min}_L &\geq 2(N - 4R)\Re(\omega^{(d-1)/2}) + R\beta_{*}^{\min},
\end{alignat}
where $\beta^{\max}_{*}$ and $\beta^{\min}_{*}$ denote the optimal values of the expression in~\Cref{eq:special_tile}.  
The reasoning parallels that of the single-special-site case, with the difference that the global optimization now splits into multiple independent ``special tiles'', one per special site.  
Since the special sites are sufficiently separated, these tiles do not overlap and can be optimized independently. Details are given in \Cref{app:local_bound}.  

Combining the results above  yields the following lower bound on the ratio:
\begin{equation} \label{eq:ratio_toric_code}
    \Lambda \geq \frac{2N + (4d-8)R}{2(N - 4R) + R\,\beta_{*}^{\max}}.
\end{equation}
In the specific case of $d=3$, the bounds are tight, and the ratio takes the explicit form
\begin{equation}
    \Lambda = \frac{N+2R}{N+2(3\cos{\pi/9}-2)R}.
\end{equation}
\Cref{fig:ratio_d=3} plots this ratio as a function of the number of special sites for a fixed total number of sites $N = 200$ (corresponding to a $10 \times 10$ grid).
\begin{figure}
    \centering
    \includegraphics[width=0.9\linewidth]{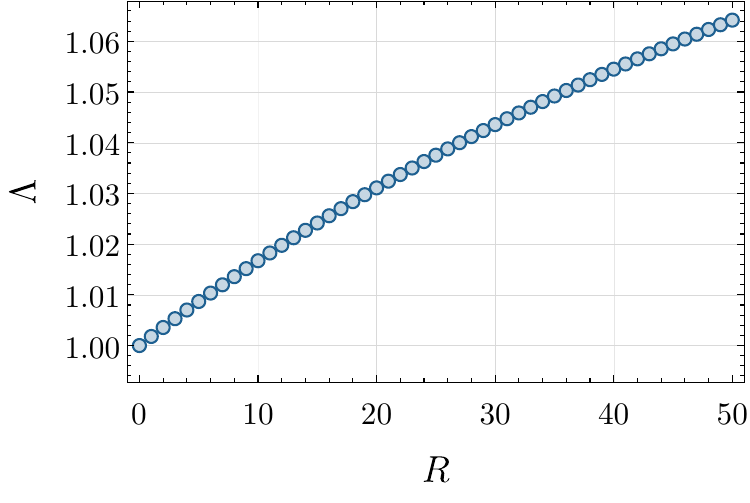}
    \caption{Ratio between the quantum and local bounds, $\Lambda$, as a function of the number of special sites, $R$.  
    The local dimension is $d = 3$, and the total number of sites is fixed to $N = 200$. The ratio increases monotonically with $R$, evidencing that a larger number of special sites enhances the robustness of the inequality against experimental imperfections.}
    \label{fig:ratio_d=3}
\end{figure}
The ratio $\Lambda$ increases with the number of special sites $R$, indicating that introducing additional special sites renders the inequality more robust against experimental imperfections, thereby improving its practical applicability.  
However, it should be noted that each additional special site requires measurements in a larger set of bases, which, on the other hand, increases experimental complexity.

Lastly, the generalized Bell expression in~\Cref{eq:bell_inequality_general} can also be employed to self-test the toric-code's ground-state subspace.  
The procedure follows the same reasoning as in the single-special-site case detailed in~\Cref{app:qutrit_self-testing}.

\section{Conclusions}\label{sec:conclusion}

In this work, we introduced a general framework for developing Bell inequalities tailored to the $\mathbb{Z}_d$ toric code for odd prime values of $d$.  
In particular, we demonstrated that the resulting inequality enables the self-testing of the $\mathbb{Z}_3$ toric code.  
As a byproduct, we also provided, to the best of our knowledge, the first self-testing proof of an entangled qutrit subspace.

Our construction of this inequality is based on carefully selecting a set of stabilizing operators of the $\mathbb{Z}_d$ toric code's ground eigenspace, along with appropriate substitutions: each operator acting on a given site is replaced by a linear combination of generalized observables. 
Each stabilizing operator $s$ then gives rise to a new term $\tilde{s}$ after substitution, and the Bell expression is obtained by summing these new terms. 
The quantum bound can be computed using a sum-of-squares decomposition, while the classical (local) bound can be estimated through enumeration techniques; in the specific case of $d=3$, this classical bound is known to be tight.

Topological quantum states play a central role in the pursuit of fault-tolerant quantum computation and error correction.  
From the standpoint of Bell nonlocality, however, such states have remained largely unexplored to date.  
By deriving Bell inequalities explicitly adapted to toric codes, here, we take a first step toward bridging this gap and enabling systematic investigations of nonlocality in topological systems.  
Moreover, our results introduce the first self-testing certification of a topologically ordered subspace, extending the reach of self-testing techniques to a new and practically relevant class of quantum states.

Naturally, we acknowledge that experimental implementation of these inequalities involves nontrivial challenges.  
The main technical task is the precise measurement of stabilizer operators, which requires sufficiently high-fidelity control to ensure that the observed correlations closely approximate their ideal values.  
Although the total number of measurement settings scales as $2N + 2d - 4$ and certain observables involve up to six parties, recent progress shows that these requirements are well within the capabilities of current platforms.  
A range of experiments have already realized the essential ingredients: preparation of qudit-based topological states, controlled creation and braiding of anyons, measurement of topological entanglement, and multi-qudit stabilizer readout at increasing scale~\cite{niu_demonstrating_2024,bluvstein_logical_2024,acharya_quantum_2025,iqbal_qutrit_2025}.  
These achievements demonstrate that the fundamental building blocks for certifying topological order in a largely device-independent manner are already in place.  
With the rapid advancements in quantum hardware and high-precision quantum simulation, only a modest step separates these existing demonstrations from a full experimental implementation of our self-testing protocol.  
Consequently, we believe that self-testing of the $\mathbb{Z}_3$ toric code is not only plausible but genuinely within reach.

Our approach also opens several promising avenues for future research.  
Since it relies on the stabilizer formalism, it can be straightforwardly extended to other stabilizer codes, such as the surface code or color codes.  
Another important open question concerns the possibility of self-testing not just the ground-state subspace but also specific logical states, potentially leveraging the intrinsic topological features of the code.

Thus, we anticipate that these results will serve as a foundation for further exploration of nonlocality and self-testing in topological quantum codes. In particular, some promising directions for future research include developing more experimentally feasible measurement protocols, investigating how topological features can enhance device-independent certification, and exploring connections with fault-tolerant quantum computation. We expect that these avenues will progressively bring device-independent certification closer to practical implementation in quantum error correction.
We also hope that the present work marks the beginning of a systematic exploration of nonlocality and self-testing in topological quantum codes, bringing device-independent certification closer to the domain of quantum error correction and fault-tolerant quantum computation.

\begin{acknowledgments}
E.V. gratefully acknowledges Tim Seynnaeve for insightful discussions regarding the tiling into elementary tiles.
P.E. and J.T. acknowledge the support received by the Dutch National Growth Fund
(NGF), as part of the Quantum Delta NL programme. 
P.E. acknowledges the support received through the NWO-Quantum Technology programme (Grant No.~NGF.1623.23.006).
P.E. further acknowledges funding by the Carl-Zeiss-Stiftung (CZS Center QPhoton).
J.T. and E.V. acknowledge the support received from the European Union's Horizon Europe research and innovation programme through the ERC StG FINE-TEA-SQUAD (Grant No.~101040729). 
R.S. was supported by the Princeton Quantum Initiative Fellowship.
This publication is part of the `Quantum Inspire - the Dutch Quantum Computer in the Cloud' project (with project number [NWA.1292.19.194]) of the NWA research program `Research on Routes by Consortia (ORC)', which is funded by the Netherlands Organization for Scientific Research (NWO).
\end{acknowledgments}

\begin{widetext}
\appendix

\section{Notation} \label{app:glossary}
In what follows, we provide a summary of the symbols frequently used throughout the manuscript.

\begin{small}
\begin{description}
    \item[$A^{(i,j)}_{x,k}$] Random variable associated to the site $(i,j)$, with input $x$ and Fourier mode $k$. Defined in \Cref{eq:expval_generalized_observables}.
    See \Cref{subsec:local_bound} for properties in the local setting.
    See \Cref{subsec:quantum_bound} for properties in the quantum setting.
    \item[$\bar{A}^{(i^*,j^*)}_{x,k}$] Linear combination of random variables. Defined in \Cref{eq:substitutions}. The constructed random variable is associated to the site $(i^*,j^*)$, with input $x$ and Fourier mode $k$.
    \item[$E_{(i^*,j^*)}(x)$] Stabilizing operator associated to a special site $(i^*,j^*)$. Defined in \Cref{eq:extra_operators}.
    \item[$\tilde{E}_{(i^*,j^*)}$] Expression of the random variables. The expression corresponds to the additional stabilizing operator $E_{(i^*,j^*)}$, where the Pauli operators have been substituted by random variables. Explicitely defined in \Cref{eq:E_tilde}.
    \item[$\edgeSet$] The set of all edges in the lattice. This also corresponds to the set of all sites. Defined in \Cref{subsec:toric_code}.
    \item[$(i^*,j^*)$ or $(i_r,j_r)$] Index of a special site. The special site is assumed to be on a vertical edge. When there is only one special site the notation $(i^*,j^*)$ is preferred.
    \item[$P_{(i,j)}$] Stabilizing operator associated to the vertex $(i,j)$. Defined in \Cref{eq:plaquette_and_vertex_operators}.
    \item[$\tilde{P}_{(i,j)}$] Expression of the random variables. The expression corresponds to the plaquette stabilizing operator $P_{(i,j)}$, where the Pauli operators have been substituted by random variables. Explicitely defined in \Cref{eq:P_tilde_V_tilde,eq:P_tilde}.
    \item[$\plaquetteSet$] The set of all plaquettes in the lattice. Defined in \Cref{subsec:toric_code}.
    \item[$\bar{\plaquetteSet}$] The set of all plaquettes far from a special site, i.e. $\mathcal{P} \setminus \{(i_r\pm 1,j_r) : r=1,\dots R \}$. Defined in \Cref{sec:bell_inequality_for_toric_code}.
    \item[$\mathcal{S}$] Set of stabilizing operators. Defined in \Cref{eq:custom_stabilizing_operators}.
    \item[$V_{(i,j)}$] Stabilizing operator associated to the vertex $(i,j)$. Defined in \Cref{eq:plaquette_and_vertex_operators}.
    \item[$\tilde{V}_{(i,j)}$] Expression of the random variables. The expression corresponds to the vertex stabilizing operator $V_{(i,j)}$, where the Pauli operators have been substituted by random variables. Explicitely defined in \Cref{eq:P_tilde_V_tilde,eq:V_tilde}.
    \item[$\vertexSet$] The set of all vertices. Defined in \Cref{subsec:toric_code}.
    \item[$\bar{\vertexSet}$] The set of all vertices far from a special site, i.e. $\mathcal{V} \setminus \{(i_r,j_r\pm 1) : r=1,\dots R \}$. Defined in \Cref{sec:bell_inequality_for_toric_code}.
\end{description}
\end{small}

\section{Complex coefficients} \label{app:complex_coefficients}
We recall here the complex coefficients used in \Cref{eq:substitutions}.
The coefficients were first derived in~\cite{kaniewski_maximal_2019}.
\begin{equation}
    \lambda_{k} = \left[ \epsilon_d \left( \frac{k}{d} \right) \right]^{-1} \omega^{-g(k,d)/48},
\end{equation}

where $\omega = \exp(2i\pi/d)$ is the $d$th root of unity, 
\begin{equation}
    \varepsilon_d \equiv 
    \begin{cases}
        1, & \text{if } d \equiv 1 \mod{4}, \\
        i, & \text{if } d \equiv 3 \mod{4} ,
    \end{cases}
\end{equation}
$\left( \frac{k}{d} \right)$ is the Legendre symbol (equals $1$ if $k$ is a quadratic residue modulo d and $-1$ otherwise), and 
\begin{equation}
g(k,d) = 
\begin{cases}
k[k^2 - d(d+6) + 3], & \text{if } k \equiv 0 \mod{2} \text{ and } (k + d + 1)/2 \equiv 0 \mod{2}, \\
k[k^2 - d(d - 6) + 3], & \text{if } k \equiv 0 \mod{2} \text{ and } (k + d + 1)/2 \equiv 1 \mod{2}, \\
k(k^2 + 3) + 2d^2(-5k + 3), & \text{if } k \equiv 1 \mod{4}, \\
k(k^2 + 3) + 2d^2(k + 3), & \text{if } k \equiv 3 \mod{4}.
\end{cases}
\end{equation}
The coefficient satisfies $\lambda_{k}^* = \lambda_{d-k}$ and $|\lambda_{k}| = 1$ for all $k \in \{1,\dots,d-1\}$.

\section{Explicit Bell expression}\label{app:explicit_bell_expression}
In what follows, we explicitly write the Bell expression.
The derivation is made in full generality for any number of outcomes $d$ prime and multiple special sites $R$.

Suppose there are $R$ special sites, denoted $(i_r,j_r) : r \in \{1,\dots,R \}$ (each of them lies on a vertical edge).
The special sites are sufficiently far from each other such that there is no vertex or plaquette containing more than one special site.
In this case, the general Bell inequality is given by \Cref{eq:bell_inequality_general}. 
For convenience, we rewrite the inequality below:
\begin{equation}
    \left[ \sum_{(i,j) \in \vertexSet} \expval*{\tilde{V}_{(i,j)}} + \sum_{(i,j) \in \plaquetteSet} \expval*{\tilde{P}_{(i,j)}} + 2 \sum_{r = 1}^{R}  \sum_{x = 2}^{d-1} \expval*{\tilde{E}_{(i_r,j_r)}(x)} \right] + \text{c.c.} \leq \beta.
\end{equation}
In what follows, we will write explicitly each term as function of the random variables $A_{x,k}$.
It is convenient to define the set of vertices that do not contain any special site $\bar{\vertexSet} \equiv \vertexSet \setminus \{(i_r,j_r\pm1) : r=1,\dots,R\}$ and similarly for plaquettes $\bar{\plaquetteSet} \equiv \plaquetteSet \setminus \{(i_r\pm1,j_r) : r = 1,\dots,R\}$.
First, let's explicitly write the terms that do not include any special site:
\begin{equation} \label{eq:P_tilde_V_tilde}
    \begin{split}
        \expval*{\tilde{V}_{(i,j)}} = & \expval*{A^{(i-1,j)}_{0,d-1} A^{(i,j-1)}_{0,d-1} A^{(i,j+1)}_{0,1} A^{(i+1,j)}_{0,1}}, \qquad \forall (i,j) \in \bar{\vertexSet},\\
        \expval*{\tilde{P}_{(i,j)}} = & \expval*{A^{(i-1,j)}_{1,1} A^{(i,j-1)}_{1,d-1} A^{(i,j+1)}_{1,1} A^{(i+1,j)}_{1,d-1}}, \qquad \forall (i,j) \in \bar{\plaquetteSet}. \\
    \end{split}
\end{equation}
Then the vertices below and above a special site $(i_r,j_r)$:
\begin{equation} \label{eq:V_tilde}
    \begin{split}
        \expval*{\tilde{V}_{(i_r,j_r-1)}} = & \expval*{A^{(i_r-1,j_r-1)}_{0,d-1} A^{(i_r,j_r-2)}_{0,d-1} \bar{A}^{(i_r,j_r)}_{0,1} A^{(i_r+1,j_r-1)}_{0,1}} = \frac{1}{\sqrt{d} \lambda_1} \sum_{y = 0}^{d-1} \expval*{A^{(i_r-1,j_r-1)}_{0,d-1} A^{(i_r,j_r-2)}_{0,d-1} A^{(i_r,j_r)}_{y,1} A^{(i_r+1,j_r-1)}_{0,1}},  \\
        \expval*{\tilde{V}_{(i_r,j_r+1)}} = & \expval*{A^{(i_r-1,j_r+1)}_{0,d-1} \bar{A}^{(i_r,j_r)}_{0,d-1} A^{(i_r,j_r+2)}_{0,1} A^{(i_r+1,j_r+1)}_{0,1}} = \frac{1}{\sqrt{d} \lambda_{d-1}} \sum_{y = 0}^{d-1} \expval*{A^{(i_r-1,j_r+1)}_{0,d-1} A^{(i_r,j_r)}_{y,d-1} A^{(i_r,j_r+2)}_{0,1} A^{(i_r+1,j_r+1)}_{0,1}}.
    \end{split}
\end{equation}
Similarly, the plaquettes on the left and on the right of a special site $(i_r,j_r)$:
\begin{equation} \label{eq:P_tilde}
    \begin{split}
        \expval*{\tilde{P}_{(i_r-1,j_r)}} & = \expval*{A^{(i_r-2,j_r)}_{1,1} A^{(i_r-1,j_r-1)}_{1,d-1} A^{(i_r-1,j_r+1)}_{1,1} \bar{A}^{(i_r,j_r)}_{1,d-1}} \\
        & = \frac{\omega^{-2(d-1)}}{\sqrt{d} \lambda_{d-1}} \sum_{y=0}^{d-1} \omega^{-(d-1)y} \expval*{A^{(i_r-2,j_r)}_{1,1} A^{(i_r-1,j_r-1)}_{1,d-1} A^{(i_r-1,j_r+1)}_{1,1} A^{(i_r,j_r)}_{y,d-1}}, \\
        \expval*{\tilde{P}_{(i_r+1,j_r)}} & = \expval*{\bar{A}^{(i_r,j_r)}_{1,1} A^{(i_r+1,j_r-1)}_{1,d-1} A^{(i_r+1,j_r+1)}_{1,1} A^{(i_r+2,j_r)}_{1,d-1}}\\
        & = \frac{\omega^{-2}}{\sqrt{d}\lambda_1} \sum_{y=0}^{d-1} \omega^{-y} \expval*{A^{(i_r,j_r)}_{y,1} A^{(i_r+1,j_r-1)}_{1,d-1} A^{(i_r+1,j_r+1)}_{1,1} A^{(i_r+2,j_r)}_{1,d-1}}.\\
    \end{split}
\end{equation}
Finally, the extra terms associated to the special site $(i_r,j_r)$ take the form
\begin{equation} \label{eq:E_tilde}
    \begin{split}
        \expval*{\tilde{E}(x)_{(i_r,j_r)}} & = \expval*{ A^{(i_r-1,j_r-1)}_{0,x-1} A^{(i_r,j_r-2)}_{0,x-1} \bar{A}^{(i_r,j_r)}_{x,1} A^{(i_r+1,j_r-1)}_{x,1} A^{(i_r+1,j_r+1)}_{1,x} A^{(i_r+2,j_r)}_{1,d-x} } \\
        & = \frac{\omega^{-x(x+1)}}{\sqrt{d}\lambda_1} \sum_{y=0}^{d-1} \omega^{-xy} \expval*{ A^{(i_r-1,j_r-1)}_{0,x-1} A^{(i_r,j_r-2)}_{0,x-1} A^{(i_r,j_r)}_{y,1} A^{(i_r+1,j_r-1)}_{x,1} A^{(i_r+1,j_r+1)}_{1,x} A^{(i_r+2,j_r)}_{1,d-x} },
    \end{split}
\end{equation}
for $x \in \{2,\dots,d-1\}$.

\section{Local bound}\label{app:local_bound}

In this section, we provide the proof of \Cref{thm:local_bound}. 
To this end, we first enounce and prove a lemma that will help us to prove the theorem.
We show that any \textit{tile of arbitrary shape} (also known as polyomino) of size greater than 1 can be decomposed into the five elementary tiles given in \Cref{fig:polyomino_decomposition}.
This result will help us to find an optimal strategy that minimizes the plaquette and the vertex terms in the Bell expression.
For a smoother initial exposition, the formal statement and the proof of \Cref{lem:polyomino_decomposition} may be skipped on a first reading and revisited later.
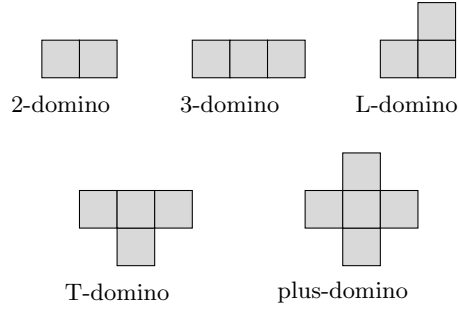
\begin{figure}[h]
\centering
\begin{tikzpicture}[scale=0.5]
\sq{0}{0}\sq{1}{0}
\node at (0.5,-0.7) {2-domino};

\begin{scope}[xshift=4cm]
\sq{0}{0}\sq{1}{0}\sq{2}{0}
\node at (1,-0.7) {3-domino};
\end{scope}

\begin{scope}[xshift=9cm]
\sq{0}{0}\sq{1}{0}\sq{1}{1}
\node at (0.7,-0.7) {L-domino};
\end{scope}

\begin{scope}[xshift=1cm,yshift=-5cm]
\sq{1}{0}\sq{0}{1}\sq{1}{1}\sq{2}{1}
\node at (1,-0.7) {T-domino};
\end{scope}

\begin{scope}[xshift=7cm,yshift=-5cm]
\sq{1}{0}\sq{0}{1}\sq{1}{1}\sq{2}{1}\sq{1}{2}
\node at (1,-0.7) {plus-domino};
\end{scope}
\end{tikzpicture}
\caption{\label{fig:polyomino_decomposition} The five elementary polyominoes used in decomposition.}
\end{figure}

\begin{lem}\label{lem:polyomino_decomposition}
Let $P$ be a \textbf{connected polyomino}, that is, a finite union of unit grid squares
such that every square can be reached from any other by a path of edge-adjacent squares.
Every connected polyomino $P$ of size $|P| \ge 2$ can be decomposed into the five
elementary tiles: the $2$-domino, $3$-domino, L-domino, T-domino, and plus-domino.
\end{lem}

The proof originated from discussions with Tim Seynnaeve\orcidlink{0000-0002-8958-6671}, to whom we express our thanks.

\begin{proof}
    We proceed by induction.  
We will show that from any polyomino $P$ with $|P|\ge 2$, one can always remove an elementary tile in such a way that the remaining polyomino is either empty, a connected polyomino $P'$ with $|P'|\ge 2$, or a collection of connected polyominoes, each of size at least 2.
 
Let $P$ be a connected polyomino with $|P| \geq 2$ and let $G$ be the adjacency graph of $P$: vertices correspond to unit squares and edges join squares sharing a side. 
Since $G$ is a finite connected graph, it contains a spanning tree $T$, and $T$ has at least one leaf. 
Let $v_1$ be a leaf of $T$, and let $v_2$ be its unique neighbour in $T$.
Note that since $P$ is connected and $|P|\ge 2$, $v_2$ necessary exists.
The local structure of the polyomino near $v_2$ determines the first tile that can be removed.  
A straightforward case analysis shows that there are 15 possible local configurations:

\textbf{Case A: $\deg_T(v_2)=1$.} 
Only $v_1,v_2$. 
Remove from $T$ the vertices $\{v_1,v_2\}$ corresponding to a 2-domino. Remaining empty graph. 

\textbf{Case B: $\deg_T(v_2)=2$.} 
The vertex $v_2$ is connected to one other vertex $v_3$.
\begin{itemize} 
    \item B1: Vertex $v_3$ is a leaf. Remove the 3-domino (or L-domino) $\{v_1,v_2,v_3\}$. Remaining: empty. 
    \item B2: Vertex $v_3$ is connected to at least one other vertex $v_3'$. Remove 2-domino $\{v_1,v_2\}$. Remaining: connected polyomino $P'$ with size $|P'| \geq 2$. 
\end{itemize} 

\textbf{Case C: $\deg_T(v_2)=3$.} 
The vertex $v_2$ is connected to two other vertices $v_3$ and $v_4$. 
\begin{itemize} 
    \item C1: $v_3$ and $v_4$ are leaves. Remove T-domino $\{v_1,v_2,v_3,v_4\}$. Remaining: empty 
    \item C2: $v_3$ is connected to at least one other vertex $v_3'$ and $v_4$ is a leave. Remove domino $\{v_1,v_2,v_4\}$. Remaining connected polyomino $P'$ with size $|P'| \geq 2$. 
    \item C2: $v_4$ is connected to at least one other vertex $v_4'$ and $v_3$ is a leave. Remove domino $\{v_1,v_2,v_3\}$. Remaining connected polyomino $P'$ with size $|P'| \geq 2$. 
    \item C4: $v_3$ is connected to at least one other vertex $v_3'$ and $v_4$ is connected to at least one other vertex $v_4'$. Remove the vertices $\{v_1,v_2\}$ corresponding to a domino. This will necessary give two disjoint tree-graphs $T'$ and $T''$. These two trees correspond to the polyominoes $P'$ and $P''$.
    $P'$ is connected and $|P'| \geq 2$. Similarly, $P''$ is connected and $|P''| \geq 2$.
\end{itemize} 

\textbf{Case D: $\deg_T(v_2)=4$.} 
The vertex $v_2$ is connected to three other vertices $v_3$ and $v_4$ and $v_5$. 
\begin{itemize} 
    \item D1: $v_3,v_4,v_5$ are leaves. Remove plus-domino $\{v_1,v_2,v_3,v_4,v_5\}$. Remaining: empty. 
    \item D2: $v_3$ is connected to another vertex $v_3'$. Vertices $v_4$ and $v_5$ are leaves. Remove T-domino $\{v_1,v_2,v_4,v_5\}$. Remaining: connected polyomino $P'$ with size $|P'| \geq = 2$. 
    \item D3: $v_4$ is connected to another vertex $v_4'$. Vertices $v_3$ and $v_5$ are leaves. Remove T-domino $\{v_1,v_2,v_3,v_5\}$. Remaining: connected polyomino $P'$ with size $|P'| \geq = 2$. 
    \item D4: $v_5$ is connected to another vertex $v_5'$. Vertices $v_3$ and $v_4$ are leaves. Remove T-domino $\{v_1,v_2,v_3,v_4\}$. Remaining: connected polyomino $P'$ with size $|P'| \geq = 2$. 
    \item D5: $v_3$ is connected to at least one other vertex $v_3'$ and $v_4$ is connected to at least one other vertex $v_4'$. Vertex $v_5$ is a leaf. Remove vertices $\{v_1,v_2,v_5\}$ corresponding to a 3-domino (or L-domino). Remaining: two disjoint tree-graphs $T'$ and $T''$. These two trees correspond to the polyominoes $P'$ and $P''$. $P'$ is connected and $|P'| \geq 2$. Similarly, $P''$ is connected and $|P''| \geq 2$. 
    \item D6: $v_3$ is connected to at least one other vertex $v_3'$ and $v_5$ is connected to at least one other vertex $v_5'$. Vertex $v_4$ is a leaf. Remove vertices $\{v_1,v_2,v_4\}$ corresponding to a 3-domino (or L-domino). Remaining: two disjoint tree-graphs $T'$ and $T''$. These two trees correspond to the polyominoes $P'$ and $P''$. $P'$ is connected and $|P'| \geq 2$. Similarly, $P''$ is connected and $|P''| \geq 2$. 
    \item D7: $v_4$ is connected to at least one other vertex $v_4'$ and $v_5$ is connected to at least one other vertex $v_5'$. Vertex $v_4$ is a leaf. Remove vertices $\{v_1,v_2,v_3\}$ corresponding to a 3-domino (or L-domino). Remaining: two disjoint tree-graphs $T'$ and $T''$. These two trees correspond to the polyominoes $P'$ and $P''$. $P'$ is connected and $|P'| \geq 2$. Similarly, $P''$ is connected and $|P''| \geq 2$. 
    \item D8: $v_3, v_4, v_5$ are connected to at least one other vertex $v_3',v_4'$ and $v_5'$ respectively. Remove vertices $\{v_1,v_2\}$ corresponding to a domino. Remaining: three disjoint tree-graphs $T'$, $T''$ and $T'''$. These three graphs correspond to three polyomino $P', P'', P'''$. Each of them is connected and has size $\geq 2$. 
\end{itemize}

In each configuration, removing the identified tile leaves connected polyominoes of size at least 2. 
By induction, each remaining polyomino can be decomposed into elementary tiles. 
Hence, any connected polyomino with $|P|\ge 2$ can be fully decomposed.
\end{proof}

We now proceed with the proof of \Cref{thm:local_bound}.
For the sake of generality, we consider the case of multiple special sites. 
The theorem, in its most general form, is stated as follows:

\begin{thm} \label{thm:local_bound_general}
    For the Bell expression in \Cref{eq:bell_inequality_general}, the maximum and minimum over all local correlation satisfy
    \begin{equation} \label{eq:local_bound_general}
    \begin{split}
        \beta^{\max}_L & \leq 2(N - 4R) + R\beta^{\max}_{*}, \\
        \beta^{\min}_L & \geq 2(N-4R)\Re(\omega^{\frac{d-1}{2}}) + R\beta^{\min}_{*},
    \end{split}
    \end{equation}
    where $\beta^{\max}_{*}$ (respectively $\beta^{\min}_{*}$) denotes the maximum (respectively the minimum) value attained by the expression $\expval*{T_{(i_*,j_*)}} + \mathrm{c.c.}$, given in \Cref{eq:special_tile}, over all possible local deterministic strategies.
    
    Furthermore, in the case of local dimension $d=3$, these bounds are tight, provided that every plaquette (respectively, vertex) not containing a special site is connected to at least one other plaquette (respectively, vertex) also not containing a special site.
\end{thm}

For simplicity, we will split the proof in two parts.
First we will prove \Cref{eq:local_bound_general} for all prime dimensions $d \geq 3$. 
And then we will restrict ourselves to the case $d = 3$, where we will prove that the bounds are tight.

\begin{proof} 
\textbf{Part 1.} 
    Every term in the expression is optimized individually over all possible LDS.
    As shown in \Cref{fig:toric-code_tiling}, the optimization can be decomposed into three categories of expressions.
    (i) All terms involving a special site, also know as special tile term and defined in \Cref{eq:special_tile}: $\expval*{T_{(i_r,j_r)}} + \text{c.c.}, \forall r \in \{1,\dots, R\}$.
    (ii) Vertex terms that do not contain a special site:
    \begin{equation}
        \expval*{\tilde{V}_{(i,j)}} + \text{c.c.} = \expval*{ A^{(i-1,j)}_{0,1} }^{d-1} \expval*{ A^{(i,j-1)}_{0,1} }^{d-1} \expval*{ A^{(i,j+1)}_{0,1} } \expval*{ A^{(i+1,j)}_{0,1}} + \text{c.c.}, \qquad \forall (i,j) \in \bar{\vertexSet} \; .
    \end{equation}
    (iii) Plaquette terms that do not contain a special site:
    \begin{equation}
        \expval*{\tilde{P}_{(i,j)}} + \text{c.c.} = \expval*{ A^{(i-1,j)}_{1,1} } \expval*{ A^{(i,j-1)}_{1,1} }^{d-1} \expval*{ A^{(i,j+1)}_{1,1} } \expval*{ A^{(i+1,j)}_{1,1}}^{d-1} + \text{c.c.},\qquad \forall \; (i,j) \in \bar{\plaquetteSet} \; .
    \end{equation}
    The sets $\bar{\plaquetteSet}$ and $\bar{\vertexSet}$ denote all plaquettes and vertices that do not contain a special site.
        
    In the case of LDS optimization, each variable $\{ \expval*{A^{(i,j)}_{x,1}} \}$ takes value in $\{1,\omega, \omega^2, \dots, \omega^{d-1}\}$.
    The maximal value of a single vertex term not containing a special site is:
    \begin{equation}
        \max_{\text{LDS}} \expval*{\tilde{V}_{(i,j)}} + \text{c.c.} = 2\max_{\text{LDS}} \Re(\expval*{\tilde{V}_{(i,j)}}) = 2
    \end{equation}
    and the minimal value is 
    \begin{equation}
        \min_{\text{LDS}} \expval*{\tilde{V}_{(i,j)}} + \text{c.c.} = 2 \min_{\text{LDS}} \Re(\expval*{\tilde{V}_{(i,j)}}) = 2 \Re(\omega^{\frac{d-1}{2}}).
    \end{equation}
    Similar results hold for plaquettes, i.e. $\max \expval*{\tilde{P}_{(i,j)}} = 2$ and $\min \expval*{\tilde{P}_{(i,j)}} = 2\Re(\omega^{\frac{d-1}{2}})$ for all $(i,j) \in \bar{\plaquetteSet}$.
    The number of plaquettes and vertices that contain no special site is $|\bar{\plaquetteSet}| + |\bar{\vertexSet}| = N-4R$, where $R$ is the number of special sites.
    By construction, $\beta_*^{\max}$ and $\beta_*^{\min}$ are the optimal values of a special tile.
    Thus the local bounds satisfy \Cref{eq:local_bound_general}.

\textbf{Part 2.}
We focus now on proving the second part of the theorem. 
From now on, we assume that the local dimension is $d=3$.
We want to show that there exist strategies that reach the upper bound and the lower bound given in \Cref{eq:local_bound_general}, i.e.
\begin{equation}\label{eq:local_bound_general_d=3}
    \beta_L^{\max} = 2N-8R + R\beta_*^{\max}, \quad \text{and} \quad \beta_L^{\min} = -N+4R + R\beta_*^{\min}.
\end{equation}
Note that in the case $d=3$, we have $\Re(\omega) = \Re(\omega^2) = -\frac{1}{2}$. 
Since there are only three types of variables $\expval*{A_{0,1}^{(i,j)}}$, $\expval*{A_{1,1}^{(i,j)}}$ and $\expval*{A_{2,1}^{(i,j)}}$, they can be visualized on three grids as shown in \Cref{fig:variables_visualization}.
Each dot is a variable.
\begin{figure}[h]
    \centering
    \includegraphics[width=0.9\linewidth]{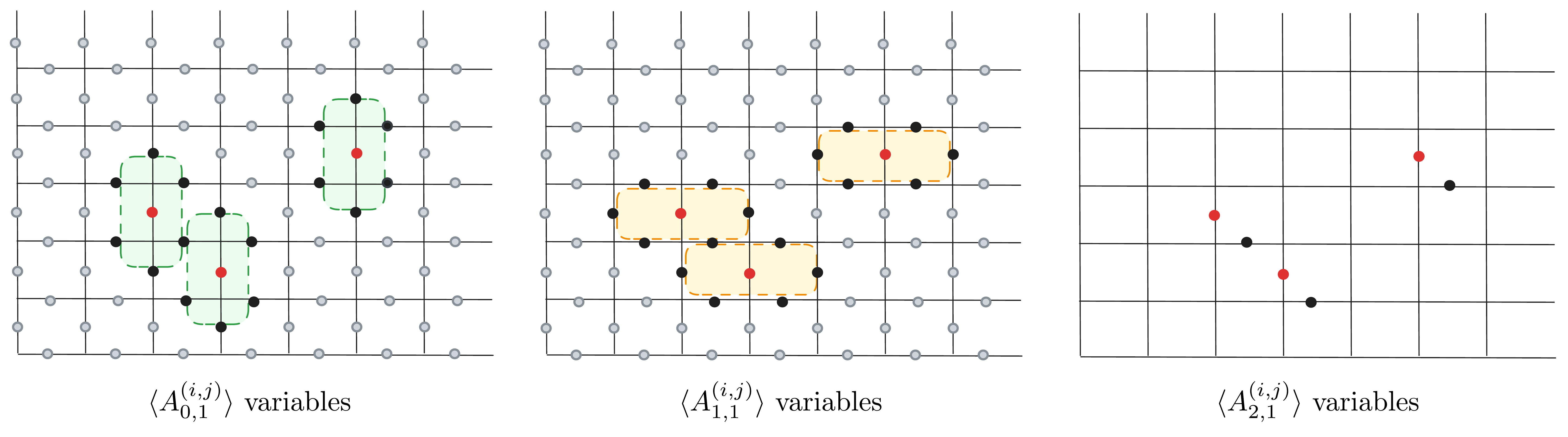}
    \caption{Visualization of the variables appearing in the Bell expression given in \Cref{eq:bell_inequality_general} in the specific case of $d=3$.
    The red dots are variables associated to special sites. The black dots are variables appearing in the expression of the special tile.}
    \label{fig:variables_visualization}
\end{figure}
For the example, we assume that $R=3$, i.e. there are three special sites (in red).
The special tile, defined in \Cref{eq:special_tile}, is thus the combination of a yellow-tile, a green tile and the corresponding $\expval*{A_{2,1}^{(i,j)}}$ variables.

The idea is to assign a value in $\{1,\omega, \omega^2\}$ to each variable such that we get the bounds in \Cref{eq:local_bound_general_d=3} when evaluating the Bell expression given by \Cref{eq:bell_inequality_general}.
In \Cref{subsubsec:local_bound_d=3}, we explain how to optimize the special tile in the case $d=3$ and we give an explicit strategy.
We recall them here:
\begin{equation}
\begin{split}
\expval*{A^{(i_r,j_r)}_{0,1}} = \omega^2, 
\quad \expval*{A^{(i_r,j_r)}_{1,1}} = 1, \quad
\expval*{A^{(i_r,j_r)}_{2,1}} = 1, \quad
\expval*{A^{(i_r+1,j_r-1)}_{2,1}} = \omega, \quad \text{(maximizing strategy)}\\
\expval*{A^{(i_r,j_r)}_{0,1}} = \omega, 
\quad \expval*{A^{(i_r,j_r)}_{1,1}} = \omega, \quad
\expval*{A^{(i_r,j_r)}_{2,1}} = \omega^2, \quad
\expval*{A^{(i_r+1,j_r-1)}_{2,1}} = 1, \quad \text{(minimizing strategy)}
\end{split}
\end{equation}
and the twelve other variables are set to $1$ for both strategies, i.e. $\expval*{A_{0,1}^{(i_r,j_r\pm2)}} = \expval*{A_{0,1}^{(i_r-1,j_r\pm1)}} = \expval*{A_{0,1}^{(i_r+1,j_r\pm1)}} = 1$ (variables on the boundary of the green tile) and $\expval*{A_{1,1}^{(i_r\pm2,j_r)}} = \expval*{A_{1,1}^{(i_r\pm1,j_r-1)}} = \expval*{A_{1,1}^{(i_r\pm1,j_r+1)}} = 1$ (variables on the boundary of the yellow tile). 

To each variable in a special tile (black and red dots in \Cref{fig:variables_visualization}), we assign the value given by the optimal strategy.
Even when special tiles touch each other, these assignments are compatible, because all the shared variables (variables at the boundaries) are equal to 1.
We are thus left with showing that there exists local strategies that satisfy
\begin{equation} \label{eq:optimize_rest}
\begin{split}
    \max_{\text{LDS}} \left[ \sum_{(i,j) \in \bar{\plaquetteSet}} \expval*{\tilde{P}_{(i,j)}} + \sum_{(i,j) \in \bar{\vertexSet}} \expval*{\tilde{V}_{(i,j)}} + \text{c.c.} \right] = 2N-8R, \\
    \min_{\text{LDS}} \left[ \sum_{(i,j) \in \bar{\plaquetteSet}} \expval*{\tilde{P}_{(i,j)}} + \sum_{(i,j) \in \bar{\vertexSet}} \expval*{\tilde{V}_{(i,j)}} + \text{c.c.} \right] = -N+4R ,
\end{split}
\end{equation}
with $\expval*{A_{0,1}^{(i,j)}} = \expval*{A_{1,1}^{(i,j)}} = 1$ for all sites at the boundary of a special tile.
The sums spans over all plaquettes (resp. vertices) that do not contain a special site, denoted $\bar{\plaquetteSet}$ (resp. $\bar{\vertexSet}$).
Visually in \Cref{fig:variables_visualization}, this is corresponds to finding an assignment for the free variables (gray dots) under the constraint that some variables (black dots on the boundary) are equal to 1.

Since there are no shared variables between plaquette terms and vertex terms, they can be optimized individually.
In what follows, we will focus only on the plaquette terms.
A similar derivation can be made for the vertex terms.
The two optimization problems are similar, thus we expect that the bounds on plaquettes only, corresponds to half the bounds in \Cref{eq:optimize_rest}.
In other word, we want to to show that there are assignments of the $\expval*{A_{1,1}^{(i,j)}}$ variables such that 
\begin{equation}\label{eq:plaquettes_optimization}
    \max_{\text{LDS}} \left[ \sum_{(i,j) \in \bar{\plaquetteSet}} \expval*{\tilde{P}_{(i,j)}} + \text{c.c.} \right] = N-4R, \qquad  
    \min_{\text{LDS}} \left[ \sum_{(i,j) \in \bar{\plaquetteSet}} \expval*{\tilde{P}_{(i,j)}} + \text{c.c.} \right] = -\frac{N}{2} + 2R,
\end{equation}
with $\expval*{A_{1,1}^{(i,j)}} = 1$ for all sites at the boundary of a special tile.

For simplicity, we call \textit{complementary tile} of plaquettes the set of plaquettes not containing a special site. 
The complementary tile of plaquettes can be visualized on the left of \Cref{fig:complementary_tile} (purple area).
When building the Bell expression in \Cref{subsec:bell_expression}, we assumed that each plaquette and each vertex contains only one special site.
In other words, green tiles and yellow tiles cannot overlap (but they can touch).
Moreover, the theorem states that the special sites must be placed such that every plaquette in the complementary tile is connected to at least another plaquette in the complementary tile.
The middle of \Cref{fig:complementary_tile} shows some forbidden configurations of the special sites.
In other words, the complementary tile is formed by tiles of arbitrary shape that have size at least 2.
\begin{figure}[h]
    \centering
    \includegraphics[width=0.9\linewidth]{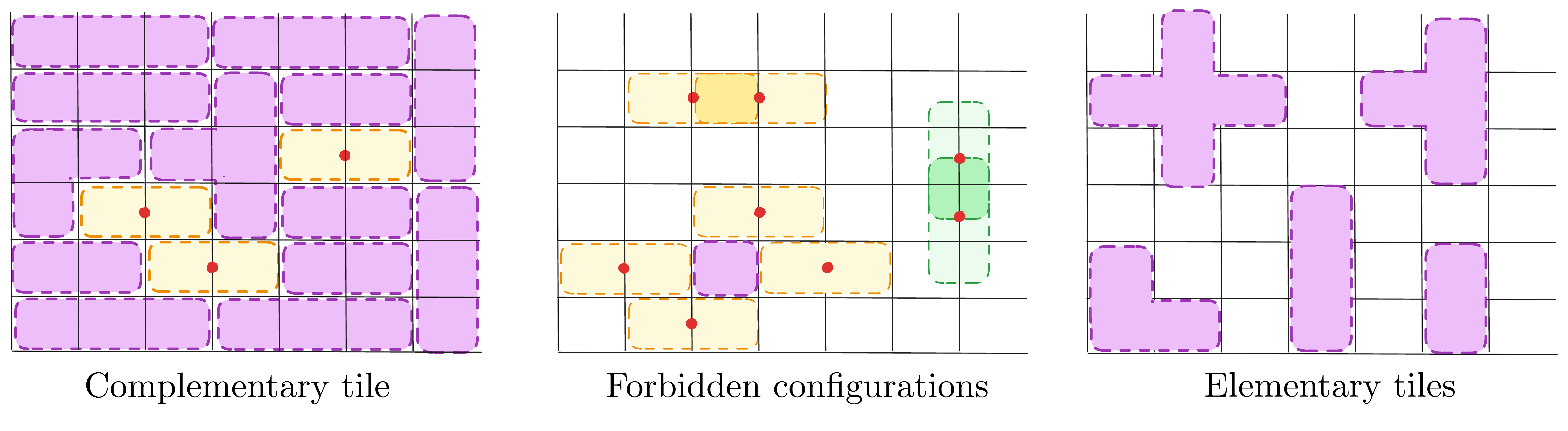}
    \caption{On the left: possible tiling of the complementary tile (purple area). In the middle: examples of forbidden configurations. The special sites (in red) must be placed in a way such that these configurations do not appear. On the right: elementary tiles used to decompose the complementary tile.}
    \label{fig:complementary_tile}
\end{figure}
Under this condition, the complementary tile can be decomposed into the five elementary tiles shown on the right of \Cref{fig:complementary_tile}. 
The proof is given in \Cref{lem:polyomino_decomposition}.
An example of such decomposition can be seen on the left of \Cref{fig:complementary_tile}.

For each elementary tile $t$, there exists strategies such that $\min_{\text{LDS}} t = -|t|$ and $\max_{\text{LDS}} t = 2|t|$, where $|t|$ is the size of the elementary tile.
The minimizing strategies are shown in \Cref{fig:minimizing_strategies}.
As a reminder, the expression of a plaquette is $\expval*{\tilde{P}_{(i,j)}} = \expval*{ A^{(i-1,j)}_{1,1} } \expval*{ A^{(i,j-1)}_{1,1} }^{2} \expval*{ A^{(i,j+1)}_{1,1} } \expval*{ A^{(i+1,j)}_{1,1}}^{2}$.
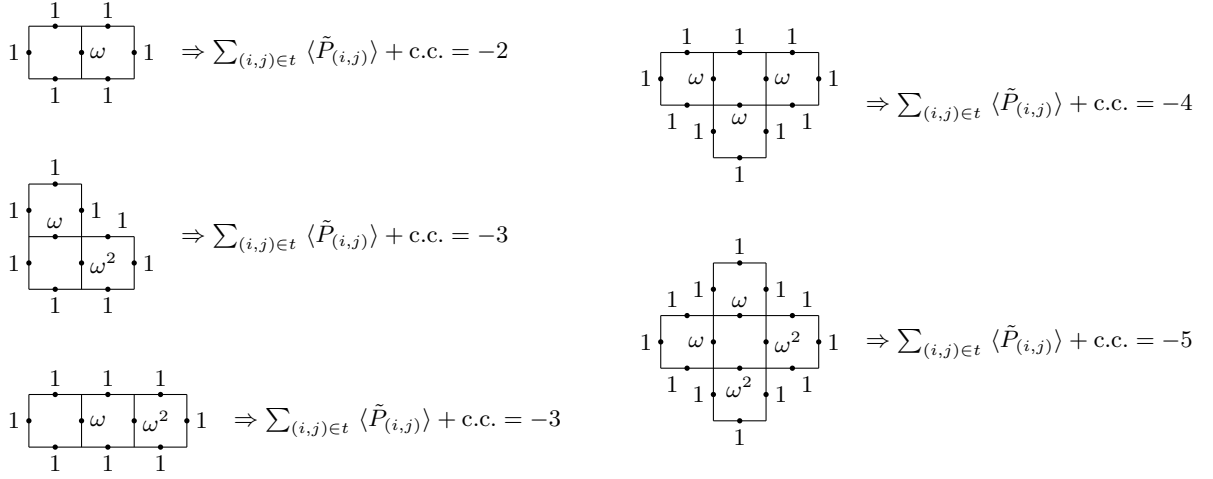
\begin{figure}[h]
    \centering
        \begin{tikzpicture}[scale=0.7]
            \begin{scope}[xshift=0cm,yshift=0cm]
                \coordinate (a0) at (0,0);
                \coordinate (a1) at (0,1);
                \coordinate (b0) at (1,0);
                \coordinate (b1) at (1,1);
                \coordinate (c0) at (2,0);
                \coordinate (c1) at (2,1);
                \coordinate (A0) at ($(a0)!0.5!(b0)$);
                \coordinate (B0) at ($(b0)!0.5!(c0)$);
                \coordinate (A1) at ($(a0)!0.5!(a1)$);
                \coordinate (B1) at ($(b0)!0.5!(b1)$);
                \coordinate (C1) at ($(c0)!0.5!(c1)$);
                \coordinate (A2) at ($(a1)!0.5!(b1)$);
                \coordinate (B2) at ($(b1)!0.5!(c1)$);
                \draw (a0)--(b0)--(c0);
                \draw (a1)--(b1)--(c1);
                \draw (a0)--(a1);
                \draw (b0)--(b1);
                \draw (c0)--(c1);
                \foreach \p in {A0,B0,A1,B1,C1,A2,B2} \fill (\p) circle (0.05);
                \node[below] at (A0) {1};
                \node[below] at (B0) {1};
                \node[left] at (A1) {1};
                \node[right] at (B1) {$\omega$};
                \node[right] at (C1) {1};
                \node[above] at (A2) {1};
                \node[above] at (B2) {1};
                \node at (6,0.5) {$\Rightarrow \sum_{(i,j) \in t} \expval*{\tilde{P}_{(i,j)}} + \text{c.c.}=-2$};
            \end{scope}
            \begin{scope}[xshift=0cm,yshift=-7cm]
                \coordinate (a0) at (0,0);
                \coordinate (a1) at (0,1);
                \coordinate (b0) at (1,0);
                \coordinate (b1) at (1,1);
                \coordinate (c0) at (2,0);
                \coordinate (c1) at (2,1);
                \coordinate (d0) at (3,0);
                \coordinate (d1) at (3,1);
                \coordinate (A0) at ($(a0)!0.5!(b0)$);
                \coordinate (B0) at ($(b0)!0.5!(c0)$);
                \coordinate (C0) at ($(c0)!0.5!(d0)$);
                \coordinate (A1) at ($(a0)!0.5!(a1)$);
                \coordinate (B1) at ($(b0)!0.5!(b1)$);
                \coordinate (C1) at ($(c0)!0.5!(c1)$);
                \coordinate (D1) at ($(d0)!0.5!(d1)$);
                \coordinate (A2) at ($(a1)!0.5!(b1)$);
                \coordinate (B2) at ($(b1)!0.5!(c1)$);
                \coordinate (C2) at ($(c1)!0.5!(d1)$);
                \draw (a0)--(b0)--(c0)--(d0);
                \draw (a1)--(b1)--(c1)--(d1);
                \draw (a0)--(a1);
                \draw (b0)--(b1);
                \draw (c0)--(c1);
                \draw (d0)--(d1);
                \foreach \p in {A0,B0,C0,A1,B1,C1,D1,A2,B2,C2} \fill (\p) circle (0.05);
                \node[below] at (A0) {1};
                \node[below] at (B0) {1};
                \node[below] at (C0) {1};
                \node[left] at (A1) {1};
                \node[right] at (B1) {$\omega$};
                \node[right] at (C1) {$\omega^2$};
                \node[right] at (D1) {1};
                \node[above] at (A2) {1};
                \node[above] at (B2) {1};
                \node[above] at (C2) {1};
                \node at (7,0.5) {$\Rightarrow \sum_{(i,j) \in t} \expval*{\tilde{P}_{(i,j)}} + \text{c.c.}=-3$};
            \end{scope}
            \begin{scope}[xshift=0cm,yshift=-4cm]
                \coordinate (a0) at (0,0);
                \coordinate (a1) at (0,1);
                \coordinate (a2) at (0,2);
                \coordinate (b0) at (1,0);
                \coordinate (b1) at (1,1);
                \coordinate (b2) at (1,2);
                \coordinate (c0) at (2,0);
                \coordinate (c1) at (2,1);
                \coordinate (A0) at ($(a0)!0.5!(b0)$);
                \coordinate (B0) at ($(b0)!0.5!(c0)$);
                \coordinate (A1) at ($(a0)!0.5!(a1)$);
                \coordinate (B1) at ($(b0)!0.5!(b1)$);
                \coordinate (C1) at ($(c0)!0.5!(c1)$);
                \coordinate (A2) at ($(a1)!0.5!(b1)$);
                \coordinate (B2) at ($(b1)!0.5!(c1)$);
                \coordinate (A3) at ($(a1)!0.5!(a2)$);
                \coordinate (B3) at ($(b1)!0.5!(b2)$);
                \coordinate (A4) at ($(a2)!0.5!(b2)$);
                \draw (a0)--(b0)--(c0);
                \draw (a1)--(b1)--(c1);
                \draw (a0)--(a1)--(a2);
                \draw (b0)--(b1)--(b2);
                \draw (c0)--(c1);
                \draw (a2)--(b2);
                \foreach \p in {A0,B0,A1,B1,C1,A2,B2, A3,B3,A4} \fill (\p) circle (0.05);
                \node[below] at (A0) {1};
                \node[below] at (B0) {1};
                \node[left] at (A1) {1};
                \node[right] at (B1) {$\omega^2$};
                \node[right] at (C1) {1};
                \node[above] at (A2) {$\omega$};
                \node[above right] at (B2) {1};
                \node[left] at (A3) {1};
                \node[right] at (B3) {1};
                \node[above] at (A4) {1};
                \node at (6,1) {$\Rightarrow \sum_{(i,j) \in t} \expval*{\tilde{P}_{(i,j)}} + \text{c.c.}=-3$};
            \end{scope}
            \begin{scope}[xshift=12cm,yshift=-1.5cm]
                \coordinate (a1) at (0,1);
                \coordinate (a2) at (0,2);
                \coordinate (b0) at (1,0);
                \coordinate (b1) at (1,1);
                \coordinate (b2) at (1,2);
                \coordinate (c0) at (2,0);
                \coordinate (c1) at (2,1);
                \coordinate (c2) at (2,2);
                \coordinate (d1) at (3,1);
                \coordinate (d2) at (3,2);
                \draw (b0)--(c0);
                \draw (a1)--(b1)--(c1)--(d1);
                \draw (a2)--(b2)--(c2)--(d2);
                \draw (a1)--(a2);
                \draw (b0)--(b1)--(b2);
                \draw (c0)--(c1)--(c2);
                \draw (d1)--(d2);
                \coordinate (bc0) at ($(b0)!0.5!(c0)$);
                \coordinate (ab1) at ($(a1)!0.5!(b1)$);
                \coordinate (bc1) at ($(b1)!0.5!(c1)$);
                \coordinate (cd1) at ($(c1)!0.5!(d1)$);
                \coordinate (ab2) at ($(a2)!0.5!(b2)$);
                \coordinate (bc2) at ($(b2)!0.5!(c2)$);
                \coordinate (cd2) at ($(c2)!0.5!(d2)$);
                \coordinate (a12) at ($(a1)!0.5!(a2)$);
                \coordinate (b01) at ($(b0)!0.5!(b1)$);
                \coordinate (b12) at ($(b1)!0.5!(b2)$);
                \coordinate (c01) at ($(c0)!0.5!(c1)$);
                \coordinate (c12) at ($(c1)!0.5!(c2)$);
                \coordinate (d12) at ($(d1)!0.5!(d2)$);
                \foreach \p in {bc0,ab1,bc1,cd1,ab2,bc2,cd2,a12,b01,b12,c01,c12,d12} \fill (\p) circle (0.05);
                \node[below] at (bc0) {1};
                \node[left] at (b01) {1};
                \node[right] at (c01) {1};
                \node[below left] at (ab1) {1};
                \node[below] at (bc1) {$\omega$};
                \node[below right] at (cd1) {1};
                \node[left] at (a12) {1};
                \node[left] at (b12)  {$\omega$};
                \node[right] at (c12) {$\omega$};
                \node[right] at (d12) {1};
                \node[above] at (ab2) {1};
                \node[above] at (bc2) {1};
                \node[above] at (cd2) {1};
                \node at (7,1) {$\Rightarrow \sum_{(i,j) \in t} \expval*{\tilde{P}_{(i,j)}} + \text{c.c.}=-4$};
            \end{scope}
            \begin{scope}[xshift=12cm,yshift=-6.5cm]
                \coordinate (a1) at (0,1);
                \coordinate (a2) at (0,2);
                \coordinate (b0) at (1,0);
                \coordinate (b1) at (1,1);
                \coordinate (b2) at (1,2);
                \coordinate (b3) at (1,3);
                \coordinate (c0) at (2,0);
                \coordinate (c1) at (2,1);
                \coordinate (c2) at (2,2);
                \coordinate (c3) at (2,3);
                \coordinate (d1) at (3,1);
                \coordinate (d2) at (3,2);
                \draw (b0)--(c0);
                \draw (b0)--(b1)--(b2)--(b3);
                \draw (c0)--(c1)--(c2)--(c3);
                \draw (b3)--(c3);
                \draw (a1)--(b1)--(c1)--(d1);
                \draw (a2)--(b2)--(c2)--(d2);
                \draw (a1)--(a2);
                \draw (d1)--(d2);
                \coordinate (bc0) at ($(b0)!0.5!(c0)$);
                \coordinate (ab1) at ($(a1)!0.5!(b1)$);
                \coordinate (bc1) at ($(b1)!0.5!(c1)$);
                \coordinate (cd1) at ($(c1)!0.5!(d1)$);
                \coordinate (ab2) at ($(a2)!0.5!(b2)$);
                \coordinate (bc2) at ($(b2)!0.5!(c2)$);
                \coordinate (cd2) at ($(c2)!0.5!(d2)$);
                \coordinate (bc3) at ($(b3)!0.5!(c3)$);
                \coordinate (a12) at ($(a1)!0.5!(a2)$);
                \coordinate (b01) at ($(b0)!0.5!(b1)$);
                \coordinate (b12) at ($(b1)!0.5!(b2)$);
                \coordinate (b23) at ($(b2)!0.5!(b3)$);
                \coordinate (c01) at ($(c0)!0.5!(c1)$);
                \coordinate (c12) at ($(c1)!0.5!(c2)$);
                \coordinate (c23) at ($(c2)!0.5!(c3)$);
                \coordinate (d12) at ($(d1)!0.5!(d2)$);
                \foreach \p in {bc0,ab1,bc1,cd1,ab2,bc2,cd2,bc3,a12,b01,b12,b23,c01,c12,c23,d12} \fill (\p) circle (0.05);
                \node[below] at (bc0) {1};
                \node[left] at (b01) {1};
                \node[right] at (c01) {1};
                \node[below left] at (ab1) {1};
                \node[below] at (bc1) {$\omega^2$};
                \node[below right] at (cd1) {1};
                \node[left] at (a12) {1};
                \node[left] at (b12)  {$\omega$};
                \node[right] at (c12) {$\omega^2$};
                \node[right] at (d12) {1};
                \node[above left] at (ab2) {1};
                \node[above] at (bc2) {$\omega$};
                \node[above right] at (cd2) {1};
                \node[left] at (b23) {1};
                \node[right] at (c23) {1};
                \node[above] at (bc3) {1};
                \node at (7,1.5) {$\Rightarrow \sum_{(i,j) \in t} \expval*{\tilde{P}_{(i,j)}} + \text{c.c.}=-5$};
        \end{scope}
        \end{tikzpicture}
    \caption{Local deterministic strategy that minimizes every elementary tile. The minimal value of the Bell expression restricted to the tile is equal to the number of plaquettes in the tile.}
    \label{fig:minimizing_strategies}
\end{figure}

The maximizing strategies corresponds to setting all the variables in the elementary tile $t$ to $1$.
In both cases, the optimal strategies are all compatible with each other (i.e. all the variables at the boundaries are equal to 1).
Putting together all the elementary tiles composing the complementary tile, we get:
\begin{equation}
    \min_{\text{LDS}} \sum_{(i,j) \in \bar{\plaquetteSet}} \expval*{\tilde{P}_{(i,j)}} + \text{c.c.} = -|\bar{\plaquetteSet}|, \quad \text{and} \quad \max_{\text{LDS}} \sum_{(i,j) \in \bar{\plaquetteSet}} \expval*{\tilde{P}_{(i,j)}} + \text{c.c.} = 2|\bar{\plaquetteSet}|.
\end{equation}
Since the complementary tile contains $|\bar{\plaquetteSet}| = \frac{N}{2}-2R$ plaquettes, the strategies given above satisfy \Cref{eq:plaquettes_optimization}.
By repeating this procedure for the vertices, we construct a strategy that reaches the maximal and minimal bounds given in \Cref{thm:local_bound_general}.
\end{proof}

\FloatBarrier
\section{Quantum bound} \label{app:quantum_bound}
In this section, we give the proof of \Cref{thm:quantum_bound}.
For the sake of generality, we state the theorem in the case of multiple special sites.
\begin{thm} \label{thm:quantum_bound_general}
For the Bell expression in \Cref{eq:bell_inequality_general}, the supremum over all quantum correlations satisfies
\begin{equation}
    \beta^{\max}_Q = 2N+(4d-8)R,
\end{equation}
where $N$ is the number of parties, $d$ is the number of outcomes and $R$ is the number of special sites in the Bell expression.
\end{thm}

The first part of the proof focuses on defining generalized quantum observables an some of their properties.
These properties will then be used to construct a sum-of-squares for the Bell expression in \Cref{eq:bell_inequality}.

\subsection{Properties of generalized observables}
As explained in the main text (see \Cref{subsec:quantum_bound}), we can assume without loss of generality that, in the quantum case, projective measurements are performed.
This implies that the random variables $A_{x,k}$ can be associated to unitary operators, called generalized observables, which satisfy $A_{x,k} = (A_{x,1})^k$.

Using the property $|\lambda_k| = 1$ (see \Cref{app:complex_coefficients}), we have:
\begin{equation} \label{eq:observable_property_1}
    \sum_{x=0}^{d-1} \bar{A}_{x,k}^{\dagger}\bar{A}_{x,k} = \frac{1}{d} \sum_{y,y'=0}^{d-1} \; \underbrace{\sum_{x=0}^{d-1} \omega^{kx(y-y')}} _{d\delta_{y,y'}} \; A_{y',k}^{\dagger} A_{y,k} = d\mathbb{I} \;.
\end{equation}
Moreover, using the properties $A_{x,k}^{\dagger} = A_{x,d-k}$ (see \Cref{subsec:quantum_bound}) and $\lambda_k^* = \lambda_{d-k}$ (see \Cref{app:complex_coefficients}), we get
\begin{equation} \label{eq:observable_property_2}
    \bar{A}_{x,k}^{\dagger} = \frac{(\omega^{-kx(x+1)})^*}{\sqrt{d} \lambda_k^*} \sum_{y=0}^{d-1} (\omega^{-kxy})^* A_{y,k}^{\dagger} = \bar{A}_{x,d-k} \;.
\end{equation}
This implies that $\bar{A}_{x,d-k} \bar{A}_{x,d-k}^{\dagger} = \bar{A}_{x,k}^{\dagger}\bar{A}_{x,k}$.

\subsection{Sum-of-squares}
With the help of \Cref{eq:observable_property_1,eq:observable_property_2}, it is possible to express the bell expression as a SOS.
Suppose there are $R$ patches, whose central site are at position $\{ (i_r,j_r) : r = 1,\dots,R \}$ (each of them lies on a vertical edge).
The Bell operator is 
\begin{equation}\label{eq:bell_operator}
    \mathcal{B} = \left[ \sum_{(i,j) \in \vertexSet} \tilde{V}_{(i,j)} + \sum_{(i,j) \in \plaquetteSet} \tilde{P}_{(i,j)} + 2 \sum_{r = 1}^{R}  \sum_{x = 2}^{d-1} \tilde{E}_{(i_r,j_r)}(x) \right] + \text{h.c.} .
\end{equation}
For convenience, we define the set of vertices that are not in a patch $\bar{\vertexSet} \equiv \vertexSet \setminus \{(i_r,j_r\pm1) : r=1,\dots,R\}$ and similarly for plaquettes $\bar{\plaquetteSet} \equiv \plaquetteSet \setminus \{(i_r\pm1,j_r) : r = 1,\dots,R\}$.
The number of such vertices/plaquettes is $|\bar{\vertexSet}| = |\bar{\plaquetteSet}| = N/2-2R$.

Now we can write the following sum-of-squares:
\begin{equation} \label{eq:general_sum-of-squares}
    \begin{split}
        \sum_{r=1}^R \left[ \; |\mathbb{I} - \tilde{V}_{(i_r,j_r-1)}|^2 + |\mathbb{I} - \tilde{V}_{(i_r,j_r+1)}^{\dagger}|^2 + |\mathbb{I}-\tilde{P}_{(i_r-1,j_r)}^{\dagger}|^2 + |\mathbb{I}-\tilde{P}_{(i_r+1,j_r)}|^2 + 2 \sum_{x=2}^{d-1} |\mathbb{I}-\tilde{E}_{(i_r,j_r)}(x)|^2 \right]\\ 
        + \sum_{(i,j) \in \bar{\vertexSet}} |\mathbb{I} - \tilde{V}_{(i,j)}|^2 + \sum_{(i,j) \in \bar{\plaquetteSet}} |\mathbb{I} - \tilde{P}_(i,j)|^2,
    \end{split}
\end{equation}
where $|M|^2 = M^{\dagger} M$.
Let's expand the expression:
\begin{equation}
\begin{split}
\sum_{r = 1}^R \bigg[ & (\mathbb{I} - \tilde{V}_{(i_r,j_r-1)})^{\dagger}(\mathbb{I} - \tilde{V}_{(i_r,j_r-1)})
    + (\mathbb{I} - \tilde{V}_{(i_r,j_r+1)}^{\dagger})^{\dagger}(\mathbb{I} - \tilde{V}_{(i_r,j_r+1)}^{\dagger}) \\
    & + (\mathbb{I} - \tilde{P}_{(i_r-1,j_r)}^{\dagger})^{\dagger} (\mathbb{I} - \tilde{P}_{(i_r-1,j_r)}^{\dagger})
    + (\mathbb{I} - \tilde{P}_{(i_r+1,j_r)})^{\dagger}(\mathbb{I} - \tilde{P}_{(i_r+1,j_r)}) \\
    & + 2 \sum_{x=2}^{d-1} (\mathbb{I}-\tilde{E}_{(i_{r},j_{r})}(x))^{\dagger}(\mathbb{I}-\tilde{E}_{(i_{r},j_{r})}(x)) \bigg] \\
    & + \sum_{{i,j} \in \bar{\vertexSet} } (\mathbb{I}-\tilde{V}_{(i,j)})^{\dagger}(\mathbb{I}-\tilde{V}_{(i,j)}) 
    + \sum_{{i,j} \in \bar{\plaquetteSet} } (\mathbb{I}-\tilde{P}_{(i,j)})^{\dagger}(\mathbb{I}-\tilde{P}_{(i,j)}) \\
    = \sum_{r=1}^R \bigg[ & \left( \mathbb{I} -  \tilde{V}_{(i_r,j_r-1)} - \tilde{V}_{(i_r,j_r-1)}^{\dagger} + [\bar{A}^{(i_r,j_r)}_{0,1}]^{\dagger} \bar{A}^{(i_r,j_r)}_{0,1} \right) + \left( \mathbb{I} -  \tilde{V}_{(i_r,j_r+1)} - \tilde{V}_{(i_r,j_r+1)}^{\dagger} + [\bar{A}^{(i_r,j_r)}_{0,1}]^{\dagger} \bar{A}^{(i_r,j_r)}_{0,1} \right) \\
    & + \left( \mathbb{I} -  \tilde{P}_{(i_r-1,j_r)} - \tilde{P}_{(i_r-1,j_r)}^{\dagger} + [\bar{A}^{(i_r,j_r)}_{0,1}]^{\dagger} \bar{A}^{(i_r,j_r)}_{0,1} \right) + \left( \mathbb{I} -  \tilde{P}_{(i_r+1,j_r)} - \tilde{P}_{(i_r+1,j_r)}^{\dagger} + [\bar{A}^{(i_r,j_r)}_{0,1}]^{\dagger} \bar{A}^{(i_r,j_r)}_{0,1} \right) \\
    & +2 \sum_{x=2}^{d-1} \left( \mathbb{I} - \tilde{E}_{(i_r,j_r)}(x) - (\tilde{E}_{(i_r,j_r)}(x))^{\dagger} + [\bar{A}^{(i_r,j_r)}_{x,1}]^{\dagger} \bar{A}^{(i_r,j_r)}_{x,1} \right) \bigg]\\
    & + \sum_{(i,j) \in \bar{\vertexSet}} (2 \mathbb{I} - \tilde{V}_{(i,j)} - \tilde{V}_{(i,j)}^{\dagger}) + \sum_{(i,j) \in \bar{\plaquetteSet}} (2 \mathbb{I} - \tilde{P}_{(i,j)} - \tilde{P}_{(i,j)}^{\dagger}) \\
    = \sum_{r=1}^R \bigg[ & 4 \mathbb{I} + 2(d-2)\mathbb{I} + 2 \underbrace{\sum_{x = 0}^{d-1} [\bar{A}^{(i_r,j_r)}_{x,1}]^{\dagger} \bar{A}^{(i_r,j_r)}_{x,1}}_{d \mathbb{I}} \bigg] + 4(N/2-2R)\mathbb{I}
    - \left[ \left( \sum_{(i,j) \in \vertexSet} \tilde{V}_{(i,j)} + \sum_{(i,j) \in \plaquetteSet} \tilde{P}_{(i,j)} - 2 \sum_{r=1}^R \sum_{x=2}^{d-1} \tilde{E}_{i_r,j_r}(x) \right) + \text{h.c.} \right] \\
    = & [2N + (4d-8)R] \mathbb{I} - \mathcal{B},
\end{split}
\end{equation}
where we used the following identities:
\begin{itemize}
    \item $\tilde{V}_{(i,j)}^{\dagger}\tilde{V}_{(i,j)} = \mathbb{I}, \quad \forall_{(i,j) \in \bar{\vertexSet}},\quad$ and $\quad\tilde{P}_{(i,j)}^{\dagger}\tilde{P}_{(i,j)} = \mathbb{I}, \quad \forall_{(i,j) \in \bar{\plaquetteSet}}$
    \item $\tilde{V}_{(i_r,j_r+1)}\tilde{V}_{(i_r,j_r+1)}^{\dagger} =  \bar{A}^{(i_r,j_r)}_{0,d-1} [\bar{A}^{(i_r,j_r)}_{0,d-1}]^{\dagger} = [\bar{A}^{(i_r,j_r)}_{0,1}]^{\dagger} \bar{A}^{(i_r,j_r)}_{0,1} = \tilde{V}_{(i_r,j_r-1)}^{\dagger}\tilde{V}_{(i_r,j_r-1)} \quad \forall_{r \in \{1,\dots,R\}}$
    \item $\tilde{P}_{(i_r-1,j_r)} \tilde{P}_{(i_r-1,j_r)}^{\dagger} = \bar{A}^{(i_r,j_r)}_{1,d-1}[\bar{A}^{(i_r,j_r)}_{1,d-1}]^{\dagger} = [\bar{A}^{(i_r,j_r)}_{1,1}]^{\dagger}\bar{A}^{(i_r,j_r)}_{1,1} = \tilde{P}_{(i_r+1,j_r)}^{\dagger}\tilde{P}_{(i_r+1,j_r)} \quad \forall_{r \in \{1,\dots,R\}}$
    \item $[\tilde{E}_{(i_r,j_r)}(x)]^{\dagger} \tilde{E}_{(i_r,j_r)}(x)  = [\bar{A}^{(i_r,j_r)}_{x,1}]^{\dagger} \bar{A}^{(i_r,j_r)}_{x,1} \quad \forall_{x \in \{2,\dots,d-1\}}$
\end{itemize}
These identities are a consequence from the fact that the generalized observables $A_{x,k}$ are assumed to be unitary.

From the SOS expression, one can see that the bound is saturated when $\expval*{\tilde{V}_{(i,j)}} = \expval*{\tilde{P}_{(i,j)}} = \expval*{\tilde{E}(x)_{(i,j)}} = 1$.
This is the case for any state in the ground-state sub-space of the toric code and measurements defined as 
\begin{equation} 
\begin{split}
    & A_{x,k}^{(i^*,j^*)} = \frac{\lambda_k}{\sqrt{d}} \sum_{y=0}^{d-1} \omega^{kxy}\omega^{ky(y+1)} (X^{1-y}Z^{y})^k, \\
    & A^{(i^*+1,j^*-1)}_{x,k} = (X^{1-x}Z^{-x})^k, \\
    & A^{(i,j)}_{0,k} = X^k \text{ and } A^{(i,j)}_{1,k} = Z^{k} .
\end{split}
\end{equation}
Note that this corresponds to the inverse map of \Cref{eq:substitutions}.

\subsection{Qutrit toric code self-testing} \label{app:qutrit_self-testing}
In this section we give the proof of \Cref{thm:self-testing}.
To this end, we first need to introduce a result by~\cite[Proposition B.2.]{kaniewski_maximal_2019} stated below as a lemma.
\begin{lem}\label{lem: self-testing}
Let $B_{0}, B_{1} \in \mathcal{B}(\mathcal{H}_{B} )$ be unitary operators satisfying $B^{3}_{0} = B^{3}_{1} = \mathbb{I}$. If the anticommutator $\{B_{0}, B_{1}\}$ is unitary, then $\mathcal{H}_{B} \equiv \mathcal{H}_{B'} \otimes \mathcal{H}_{B''}$ for $\mathcal{H}_{B'} \equiv \mathbb{C}^{3}$ and there exists a unitary $U :\mathcal{H}_{B} \rightarrow \mathcal{H}_{B'} \otimes \mathcal{H}_{B''}$ such that
\begin{equation}
UB_{0}U^{\dagger}=X\otimes \mathbb{I},\quad UB_{1}U^{\dagger} = Z\otimes Q_{1}+Z^{2}\otimes Q_{2},
\end{equation}
where $Q_{1},Q_{2}\in\mathcal{B}(\mathcal{H}_{B''})$ are orthogonal projectors satisfying $Q_{1}+Q_{2}=\mathbb{I}$.
\end{lem} 
Note that the exact form of the operators is different than the one in~\cite[Proposition B.2.]{kaniewski_maximal_2019}, however, as these are unitarily equivalent to each other, we pick the ones more suitable for our purposes.

We now proceed with the proof of \Cref{thm:self-testing}. For the sake of generality, we consider the case of multiple special
sites. 
The theorem, in its most general form, is stated as follows:

\begin{thm}\label{thm:self-testing_general}
Let $\ket{\psi_{\operatorname{max}}}\in\mathcal{H}$ be a state achieving maximal violation of Eq. \eqref{eq:bell_inequality_general} for $d=3$, where we assume that a partial trace of $\ket{\psi_{\operatorname{max}}}$ over any subset of parties gives a full-rank density matrix. Then there exists a unitary $U=\bigotimes_{(i,j)\in\edgeSet}U_{(i,j)}$ such that $\mathcal{H}=\mathcal{H}'\otimes\mathcal{H}''$ and
\begin{equation}
U \ket{\psi_{\max}} =\sum_{k=1}^{4} \alpha_{1,k}\ket{\tau_{k}}\ket{\eta_{1,k}}+\alpha_{2,k}\ket{\tau_{k}^{*}}\ket{\eta_{2,k}},
\end{equation}
where $\alpha_{l,k}\in \mathbb{C}$ are unknown coefficient satisfying $\sum_{l=1}^{2}\sum_{k=1}^{4}|\alpha_{l,k}|^{2}=1$, $\{\ket{\tau{_k}}\}_{k=1}^{4}\in\mathcal{H}'$ is an orthonormal basis of a qutrit toric code, $\ket{\tau_{k}^{*}}$ is a complex conjugation of $\ket{\tau_{k}}$, and $\ket{\eta_{l,k}}\in\mathcal{H}''$ are axillary states. 
\end{thm}

\begin{proof}
Let us begin by outlining the proof plan. The general strategy is to show that from the maximal violation of the inequality \eqref{eq:bell_inequality_general} it follows that the observables satisfy the conditions of Lemma \ref{lem: self-testing}. 
This then allows us to transform the operators making up the Bell functional into stabilizing operators of a qutrit toric code, which then proves its self-testing. 
For clarity, the proof is divided into four parts: \textit{Part I} and \textit{II} are dedicated to proving conditions necessary for the use of Lemma \ref{lem: self-testing}: in \textit{Part I}, we show that all observables satisfy $A^{3}=\mathbb{I}$ and in \textit{Part II} we prove that all anticommutators of observables acting on the same party are unitary. 

This then allows us to use Lemma \ref{lem: self-testing} to transform the observables into the Pauli operators times some auxiliary operators. 
However, the direct application of Lemma \ref{lem: self-testing} will not give us the desirable stabilizing operators, due to the presence of projectors $Q_{1}$ and $Q_{2}$. 
To solve this issue, in \textit{Part III} we show that any terms containing the tensor product of at least one $Q_{1}$ and one $Q_{2}$ can be disregarded as they do not give any contribution to the maximal violation of the inequality. 
Lastly, in \textit{Part IV} we show how the transformed operators allow us to prove the self-testing of a qutrit toric code.

\textit{Part I.}
Since $\ket{\psi_{\text{max}}}$ achieves maximal violation, the expected value of the sum-of-squares (SOS) in \Cref{eq:general_sum-of-squares} evaluated over the state $\ket{\psi_{\max}}$ gives $0$.
As each term in the SOS is positive, for the whole sum to be $0$, each individual term has to be $0$. Moreover, for any matrix $M$ and a state $\ket{\psi}$ we have $\bra{\psi}|M|^{2}\ket{\psi}=0 \Rightarrow M\ket{\psi}=0$, therefore from the SOS \eqref{eq:general_sum-of-squares} we have
\begin{equation}
\begin{aligned}
 \tilde{V}_{(i_r,j_r-1)} \ket{\psi_{\text{max}}} &= \ket{\psi_{\text{max}}},\qquad \tilde{V}_{(i_r,j_r+1)}^{\dagger} \ket{\psi_{\text{max}}} = \ket{\psi_{\text{max}}},\qquad
 \tilde{P}_{(i_r-1,j_r)}^{\dagger} \ket{\psi_{\text{max}}} = \ket{\psi_{\text{max}}},\\ 
 &\tilde{P}_{(i_r+1,j_r)} \ket{\psi_{\text{max}}} = \ket{\psi_{\text{max}}}, \qquad \tilde{E}_{(i_r,j_r)}(2) \ket{\psi_{\text{max}}} = \ket{\psi_{\text{max}}}, 
\end{aligned}
\end{equation}
for all $r\in\{1,\ldots,R\}$, and
\begin{equation}
\forall_{(i,j) \in \bar{\vertexSet}} \quad \tilde{V}_{(i,j)} \ket{\psi_{\text{max}}} = \ket{\psi_{\text{max}}}, \qquad
\forall_{(i,j) \in \bar{\plaquetteSet}} \quad \tilde{P}_{(i,j)} \ket{\psi_{\text{max}}} = \ket{\psi_{\text{max}}} .
\end{equation}
Written explicitly, this reads 
\begin{equation}
    \begin{aligned}
        &\forall_{r \in \{1,\dots,R\}}, \quad &A^{(i_r-1,j_r-1)}_{0,2} A^{(i_r,j_r-2)}_{0,2} \bar{A}^{(i_r,j_r)}_{0,1} A^{(i_r+1,j_r-1)}_{0,1} \ket{\psi_{\text{max}}} =& \ket{\psi_{\text{max}}},\\
        &\forall_{r \in \{1,\dots,R\}}, \quad &\left( A^{(i_r-1,j_r+1)}_{0,2} \bar{A}^{(i_r,j_r)}_{0,2} A^{(i_r,j_r+2)}_{0,1} A^{(i_r+1,j_r+1)}_{0,1} \right)^{\dagger} \ket{\psi_{\text{max}}} =& \ket{\psi_{\text{max}}},\\
        &\forall_{r \in \{1,\dots,R\}}, \quad &\left( A^{(i_r-2,j_r)}_{1,1} A^{(i_r-1,j_r-1)}_{1,2} A^{(i_r-1,j_r+1)}_{1,1} \bar{A}^{(i_r,j_r)}_{1,2} \right)^{\dagger} \ket{\psi_{\text{max}}} =& \ket{\psi_{\text{max}}},\\
        &\forall_{r \in \{1,\dots,R\}}, \quad &\bar{A}^{(i_r,j_r)}_{1,1} A^{(i_r+1,j_r-1)}_{1,2} A^{(i_r+1,j_r+1)}_{1,1} A^{(i_r+2,j_r)}_{1,2} \ket{\psi_{\text{max}}} =& \ket{\psi_{\text{max}}},\\
        &\forall_{r \in \{1,\dots,R\}}, \quad &A^{(i_r-1,j_r-1)}_{0,1} A^{(i_r,j_r-2)}_{0,1} \bar{A}^{(i_r,j_r)}_{2,1} A^{(i_r+1,j_r-1)}_{2,1} A^{(i+1,j+1)}_{1,2} A^{(i+2,j)}_{1,1} \ket{\psi_{\text{max}}} =& \ket{\psi_{\text{max}}}, \\
        &\forall_{(i,j) \in \bar{\vertexSet}}, \quad &A^{(i-1,j)}_{0,2} A^{(i,j-1)}_{0,2} A^{(i,j+1)}_{0,1} A^{(i+1,j)}_{0,1} \ket{\psi_{\text{max}}} =& \ket{\psi_{\text{max}}}, \\
        &\forall_{(i,j) \in \bar{\plaquetteSet}}, \quad &A^{(i-1,j)}_{1,1} A^{(i,j-1)}_{1,2} A^{(i,j+1)}_{1,1} A^{(i+1,j)}_{1,2} \ket{\psi_{\text{max}}} =& \ket{\psi_{\text{max}}}.
    \end{aligned}
\end{equation}
As was shown in \Cref{eq:observable_property_2}, we have that $A_{x,k}^{\dagger} = A_{x,d-k}$ and $\bar{A}_{x,k}^{\dagger} = \bar{A}_{x,d-k}$, which allows us to rewrite the above as 
\begin{equation} \label{eq:stab_observables}
    \begin{aligned}
        &\forall_{r \in \{1,\dots,R\}}, \quad &A^{(i_r-1,j_r-1)}_{0,2} A^{(i_r,j_r-2)}_{0,2} \bar{A}^{(i_r,j_r)}_{0,1} A^{(i_r+1,j_r-1)}_{0,1} \ket{\psi_{\text{max}}} =& \ket{\psi_{\text{max}}},\\
        &\forall_{r \in \{1,\dots,R\}}, \quad &A^{(i_r-1,j_r+1)}_{0,1} \bar{A}^{(i_r,j_r)}_{0,1} A^{(i_r,j_r+2)}_{0,2} A^{(i_r+1,j_r+1)}_{0,2} \ket{\psi_{\text{max}}} =& \ket{\psi_{\text{max}}},\\
        &\forall_{r \in \{1,\dots,R\}}, \quad &A^{(i_r-2,j_r)}_{1,2} A^{(i_r-1,j_r-1)}_{1,1} A^{(i_r-1,j_r+1)}_{1,2} \bar{A}^{(i_r,j_r)}_{1,1} \ket{\psi_{\text{max}}} =& \ket{\psi_{\text{max}}},\\
        &\forall_{r \in \{1,\dots,R\}}, \quad &\bar{A}^{(i_r,j_r)}_{1,1} A^{(i_r+1,j_r-1)}_{1,2} A^{(i_r+1,j_r+1)}_{1,1} A^{(i_r+2,j_r)}_{1,2} \ket{\psi_{\text{max}}} =& \ket{\psi_{\text{max}}},\\
        &\forall_{r \in \{1,\dots,R\}}, \quad &A^{(i_r-1,j_r-1)}_{0,1} A^{(i_r,j_r-2)}_{0,1} \bar{A}^{(i_r,j_r)}_{2,1} A^{(i_r+1,j_r-1)}_{2,1} A^{(i+1,j+1)}_{1,2} A^{(i+2,j)}_{1,1} \ket{\psi_{\text{max}}} =& \ket{\psi_{\text{max}}}, \\
        &\forall_{(i,j) \in \bar{\vertexSet}}, \quad &A^{(i-1,j)}_{0,2} A^{(i,j-1)}_{0,2} A^{(i,j+1)}_{0,1} A^{(i+1,j)}_{0,1} \ket{\psi_{\text{max}}} =& \ket{\psi_{\text{max}}}, \\
        &\forall_{(i,j) \in \bar{\plaquetteSet}}, \quad &A^{(i-1,j)}_{1,1} A^{(i,j-1)}_{1,2} A^{(i,j+1)}_{1,1} A^{(i+1,j)}_{1,2} \ket{\psi_{\text{max}}} =& \ket{\psi_{\text{max}}} .
    \end{aligned}
\end{equation}
Our first goal is to show that the operators $A^{(i,j)}_{x,k} \; \forall_{(i,j) \neq (i_r,j_r)}$ and $\bar{A}^{(i_r,j_r)}_{x,k}$ can be transformed with a single unitary to generalized Pauli matrices. 
To this end, we make use of Lemma \ref{lem: self-testing}, and so we first need to show that
\begin{equation}
\forall_{(i,j) \neq (i_r,j_r)} \; \left(A^{(i,j)}_{x,1}\right)^{d}=\mathbb{I},\qquad \qquad \forall_{r\in\{1,\ldots,R\}}\;\left(\bar{A}^{(i_r,j_r)}_{x,1}\right)^{d} = \mathbb{I}.
\end{equation}
First, as explained in \Cref{subsec:quantum_bound}, we can assume that the measurements are PVM.
We impliedly have that $(A^{(i,j)}_{x,1})^{d} = \mathbb{I}$ is satisfied.
Therefore, for this condition to be satisfied by all operators, we have to show that $(\bar{A}^{(i,j)}_{x,1})^{d} = \mathbb{I}$. 

Let us apply $d$ times the operator $O \equiv A^{(i_r-1,j_r-1)}_{0,2} A^{(i_r,j_r-2)}_{0,2} \bar{A}^{(i_r,j_r)}_{0,1} A^{(i_r+1,j_r-1)}_{0,1}$ on state $\ket{\psi_{\text{max}}}$.
On one hand we have $O^d = (\bar{A}^{(i_r,j_r)}_{0,1})^d$ and on the other hand the first equation from \Cref{eq:stab_observables} applied recursively gives $O^d \ket{\psi_{\text{max}}} = \ket{\psi_{\text{max}}}$.
Thus,
\begin{equation}
\forall_{r\in\{1,\ldots,R\}},\quad\left(\bar{A}^{(i_r,j_r)}_{0,1}\right)^{d}\ket{\psi_{\max}} = \ket{\psi_{\max}},
\end{equation}
which shows that $\left(\bar{A}^{(i_r,j_r)}_{0,1}\right)^{d}=\mathbb{I}$ on the support of $\ket{\psi_{\max}}$, and since we assume that the reduced matrix of $\ket{\psi_{\max}}$ is full rank, this implies $\left(\bar{A}^{(i_r,j_r)}_{0,1}\right)^{d}=\mathbb{I}$. 
The same procedure can be performed for any input $x\in\{0,1,2\}$ and thus we have 
\begin{equation}\label{eq:A_bar_d}
\left(\bar{A}^{(i_r,j_r)}_{x,1}\right)^{d}=\mathbb{I}, \quad \forall_{x \in \{0,1,2\}}.
\end{equation}

\textit{Part II.}
Here we prove a second condition required to use Lemma \ref{lem: self-testing}, i.e., that
\begin{equation}
\left\{A^{(i,j)}_{0,1},A^{(i,j)}_{1,1}\right\}^{\dagger}\left\{A^{(i,j)}_{0,1},A^{(i,j)}_{1,1}\right\}=\left\{A^{(i,j)}_{0,1},A^{(i,j)}_{1,1}\right\}\left\{A^{(i,j)}_{0,1},A^{(i,j)}_{1,1}\right\}^{\dagger}=\mathbb{I},
\end{equation}
holds for all $(i,j)\in\edgeSet$. To start, from construction (see Ref.~\cite[Eq. (114)]{kaniewski_maximal_2019}), we have that
\begin{equation}\label{eq: commutation_bar}
    \left\{\bar{A}^{(i_r,j_r)}_{0,1},\bar{A}^{(i_r,j_r)}_{1,1}\right\} = -\bar{A}^{(i_r,j_r)}_{2,2}.
\end{equation}
Using the same idea as for the proof of \Cref{eq:A_bar_d} it can be easily shown that $\bar{A}^{(i_r,j_r)}_{2,2}$ is unitary, then the anticommutator is unitary as well.
Therefore, by virtue of Lemma \ref{lem: self-testing} there exist a unitary $U^{(i_r,j_r)}$ such that
\begin{equation}
U^{(i_r,j_r)}\bar{A}_{0,1}^{(i_r,j_r)}\left(U^{(i_r,j_r)}\right)^{\dagger} = X\otimes \mathbb{I},\quad U^{(i_{r},j_{r})}\bar{A}_{1,1}^{(i_r,j_r)}\left(U^{(i_r,j_r)}\right)^{\dagger} = Z\otimes Q_{1}^{(i_{r},j_{r})}+Z^{2}\otimes Q_{2}^{(i_{r},j_{r})},
\end{equation}
for all $r\in\{1,\ldots,R\}$, where $U^{(i_r,j_r)}\in \mathcal{B}(\mathcal{H}_{(i_{r},j_{r})})$, $Z\in \mathcal{B}(\mathcal{H}'_{(i_{r},j_{r})})$, $Q_{1}^{(i_{r},j_{r})},Q_{2}^{(i_{r},j_{r})}\in \mathcal{B}(\mathcal{H}''_{(i_{r},j_{r})})$, and $\mathcal{H}_{(i_{r},j_{r})}=\mathcal{H}'_{(i_{r},j_{r})}\otimes\mathcal{H}''_{(i_{r},j_{r})}$

Next, we need to show that $\left\{ A_{0,1}^{(i,j)},A_{1,1}^{(i,j)}\right\}$ is unitary for all $(i,j)\neq (i_r,j_r)$. To start, let us consider the case $(i,j)=(i_r-1,j_r-1)$. 
We can make use of the first and third relation from \Cref{eq:stab_observables}, which, again using the projective measurement assumptions, can be expressed as
\begin{equation}
\begin{aligned}
 A^{(i_r-1,j_r-1)}_{0,1} \ket{\psi_{\max}}&= A^{(i_r,j_r-2)}_{0,2} \bar{A}^{(i_r,j_r)}_{0,1} A^{(i_r+1,j_r-1)}_{0,1}\ket{\psi_{\max}},\\
 A^{(i_r-1,j_r-1)}_{1,2} \ket{\psi_{\max}} &= A^{(i_r-2,j_r)}_{1,2} A^{(i_r-1,j_r+1)}_{1,2} \bar{A}^{(i_r,j_r)}_{1,1}\ket{\psi_{\max}}.
\end{aligned} 
\end{equation}
Using these relations, we can express the anticommutator of the two observables on the left-hand side as
\begin{equation}
\begin{split}
   \left\{A_{1,2}^{(i_r-1,j_r-1)}, A_{0,1}^{(i_r-1,j_r-1)}\right\}\ket{\psi_{\max}}=\Big(&A_{1,2}^{(i_r-1,j_r-1)}  A_{0,1}^{(i_r-1,j_r-1)}+ A_{0,1}^{(i_r-1,j_r-1)}A_{1,2}^{(i_r-1,j_r-1)}\Big) \ket{\psi_{\max}}\\
    =\Big(&
    A^{(i_r,j_r-2)}_{0,2} \bar{A}^{(i_r,j_r)}_{0,1} A^{(i_r+1,j_r-1)}_{0,1}A^{(i_r-2,j_r)}_{1,2} A^{(i_r-1,j_r+1)}_{1,2} \bar{A}^{(i_r,j_r)}_{1,1}\\
    & + A^{(i_r-2,j_r)}_{1,2} A^{(i_r-1,j_r+1)}_{1,2} \bar{A}^{(i_r,j_r)}_{1,1}A^{(i_r,j_r-2)}_{0,2} \bar{A}^{(i_r,j_r)}_{0,1} A^{(i_r+1,j_r-1)}_{0,1}\Big) \ket{\psi_{\max}}. 
\end{split}
\end{equation}
Taking to the front every observable that commutes with every other observable in the above expression equals
\begin{equation}
\begin{aligned}
A^{(i_r,j_r-2)}_{0,2} A^{(i_r+1,j_r-1)}_{0,1}A^{(i_r-2,j_r)}_{1,2} A^{(i_r-1,j_r+1)}_{1,2}\Big\{\bar{A}^{(i_r,j_r)}_{0,1},\bar{A}^{(i_r,j_r)}_{1,1}&\Big\}\ket{\psi_{\max}}.
\end{aligned}
\end{equation}
The first four operators are unitary by assumption, and we have already proven that $\Big\{\bar{A}^{(i_r,j_r)}_{0,1},\bar{A}^{(i_r,j_r)}_{1,1}\Big\}$ is unitary, therefore $\Big\{ A_{0,1}^{(i_{r}-1,j_{r}-1)},A_{1,1}^{(i_{r}-1,j_{r}-1)}\Big\}$ is unitary on the support of $\ket{\psi_{\max}}$. Lastly, from our assumption that the reduced matrix of $\ket{\psi_{\max}}$ is full rank follows that $\Big\{ A_{0,1}^{(i_{r}-1,j_{r}-1)},A_{1,1}^{(i_{r}-1,j_{r}-1)}\Big\}$ is unitary.

This approach can then be repeated for all of the neighbours of $(i_{r},j_{r})$, i.e., for $(i_{r}\pm 1,j_{r}\pm 1)$. Moreover, we can next use the fact that $\Big\{A^{(i_r \pm 1,j_r \pm 1)}_{0,1},A^{(i_r \pm 1,j_r \pm 1)}_{1,1}\Big\}$ are unitary to prove the unitarity of the anticomutator for the neighbours of $(i_{r}\pm1,j_{r}\pm1)$ and so on. Through this procedure, we can show that $\Big\{A^{(i,j)}_{0,1},A^{(i,j)}_{1,1}\Big\}$ is unitary for all $(i,j)\neq (i_{r},j_{r})$. 

\textit{Part III.} From two previous parts, we can conclude that by virtue of Lemma \ref{lem: self-testing}, for each party $(i,j)$ there exists a unitary $U^{(i,j)}$ such that
\begin{equation}
\begin{aligned}
U^{(i,j)}A_{0,1}^{(i,j)}\left(U^{(i,j)}\right)^{\dagger} = X\otimes \mathbb{I},\quad U^{(i,j)}A_{1,1}^{(i,j)}\left(U^{(i,j)}\right)^{\dagger} = Z\otimes Q_{1}^{(i,j)}+Z^{2}\otimes Q_{2}^{(i,j)}, \qquad \textrm{for all }(i,j)\neq(i_{r},j_{r}),\\ 
U^{(i_r,j_r)}\bar{A}_{0,1}^{(i_r,j_r)}\left(U^{(i_r,j_r)}\right)^{\dagger} = X\otimes \mathbb{I},\quad U^{(i_{r},j_{r})}\bar{A}_{1,1}^{(i_r,j_r)}\left(U^{(i_r,j_r)}\right)^{\dagger} = Z\otimes Q_{1}^{(i_{r},j_{r})}+Z^{2}\otimes Q_{2}^{(i_{r},j_{r})}.\\ 
\end{aligned}
\end{equation}
Moreover, for each Hilbert space $\mathcal{H}_{(i,j)}$ associated with the party $(i,j)$ Lemma \ref{lem: self-testing} implies that $\mathcal{H}_{(i,j)}=\mathbb{C}^{3}\otimes \mathbb{C}^{t}$ where $\dim (\mathcal{H}_{(i,j)})=3t$. Then, in the equation above, the space $\mathbb{C}^{3}$ is acted on with the generalized Pauli operators, whereas $\mathbb{C}^{t}$ is acted on with the identity operator $\mathbb{I}$ and projectors $Q_{1}^{(i,j)}, Q_{2}^{(i,j)}$.

To make the notation more readable, we combine the "nontrivial" ($\mathbb{C}^{3}$) and "trivial" ($\mathbb{C}^{t}$) Hilbert spaces
\begin{equation}
\mathcal{H}'=\bigotimes_{(i,j)\in\edgeSet}\mathcal{H}'_{(i,j)}=(\mathbb{C}^{3})^{\otimes N}, \qquad  \mathcal{H}''=\bigotimes_{(i,j)\in\edgeSet}\mathcal{H}''_{(i,j)}=\bigotimes_{(i,j)\in\edgeSet}\mathbb{C}^{t_{(i,j)}}.
\end{equation}
Using this notation let us consider operators $X_{(i,j)},Z_{(i,j)}\in \mathcal{B}\left(\mathcal{H}\right)$ that act on the subspace $\mathcal{H}_{(i,j)}$ with operators $X$ and $Z$ respectively and with identity on every other subspace composing $\mathcal{H}$. Similarly, the projector $\Pi_{(i,j)}^{k}\in \mathcal{B}(\mathcal{H})$ acts on $\mathcal{H}_{(i,j)}''$ with projector $Q^{(i,j)}_{k}$ and with $\mathbb{I}$ on every other subspace of $\mathcal{H}$.

Let us act with the unitary $\bigotimes_{(i,j)}U^{(i,j)}$ on operators and states in Eq. \eqref{eq:stab_observables} (excluding $E_{(i_{r},j_{r})}(x)$ as they are inconsequential for the rest of the proof). This transformation yields the following stabilizing relations
\begin{equation}\label{eq:stabilizing_with_omega}
\begin{aligned} 
\forall_{(i,j) \in \vertexSet}\qquad X_{(i-1,j)}^{2} X_{(i,j-1)}^{2} X_{(i,j+1)} X_{(i+1,j)}\otimes\mathbb{I} \ket{\tilde{\psi}_{\max}}&= \ket{\tilde{\psi}_{\max}},\\
\forall_{(i,j) \in \plaquetteSet}\qquad \Omega_{(i-1,j)} \Omega_{(i,j-1)}^{\dagger} \Omega_{(i,j+1)} \Omega_{(i+1,j)}^{\dagger}\ket{\tilde{\psi}_{\max}}&= \ket{\tilde{\psi}_{\max}},
\end{aligned}
\end{equation}
where $\Omega_{(i,j)}=\Big(Z_{(i,j)}\Pi_{(i,j)}^{1}+Z_{(i,j)}^{2}\Pi_{(i,j)}^{2}\Big)$, $\mathbb{I}\in\mathcal{B}\left(\mathcal{H}''\right)$, and
\begin{equation}
\ket{\tilde{\psi}_{\textrm{max}}}= \bigotimes_{(i,j)}U^{(i,j)}\ket{\psi_{\textrm{max}}}.
\end{equation}

In the next step we show that some terms in the expansion of $\Omega_{(i-1,j)} \Omega_{(i,j-1)}^{\dagger} \Omega_{(i,j+1)} \Omega_{(i+1,j)}^{\dagger}$ nullify $\ket{\tilde{\psi}_{\max}}$, and so can be removed from the stabilizing relations. To this end, let us consider the following 
\begin{equation}\label{eq:projector_derivation}
\begin{aligned} 
\Omega_{(i-1,j)} \Omega_{(i,j-1)}^{\dagger}& \Omega_{(i,j+1)} \Omega_{(i+1,j)}^{\dagger}\ket{\tilde{\psi}_{\max}}\\
&= \Omega_{(i-1,j)} \Omega_{(i,j-1)}^{\dagger} \Omega_{(i,j+1)} \Omega_{(i+1,j)}^{\dagger}\;X_{(i,j+1)}^{2} X_{(i+1,j)}^{2} X_{(i+1,j+2)} X_{(i+2,j+1)}\ket{\tilde{\psi}_{\max}}\\
&=X_{(i,j+1)}^{2} X_{(i+1,j)}^{2} X_{(i+1,j+2)} X_{(i+2,j+1)}\;\Omega_{(i-1,j)} \Omega_{(i,j-1)}^{\dagger} \tilde{\Omega}_{(i,j+1)} \tilde{\Omega}_{(i+1,j)}^{\dagger}\ket{\tilde{\psi}_{\max}},
\end{aligned}
\end{equation}
where $\tilde{\Omega}_{k,l}=\omega^{2} Z_{(k,l)}\Pi_{(k,l)}^{1}+\omega Z_{(k,l)}^{2}\Pi^{2}_{(k,l)}$. Since 
\begin{equation}\label{eq:stabilizing_one_omega}
\Omega_{(i-1,j)} \Omega_{(i,j-1)}^{\dagger} \Omega_{(i,j+1)} \Omega_{(i+1,j)}^{\dagger}\ket{\tilde{\psi}_{\max}}=\ket{\tilde{\psi}_{\max}},
\end{equation}
we have that
\begin{equation}
\begin{aligned} 
&X_{(i,j+1)}^{2} X_{(i+1,j)}^{2} X_{(i+1,j+2)} X_{(i+2,j+1)}\;\Omega_{(i-1,j)} \Omega_{(i,j-1)}^{\dagger} \tilde{\Omega}_{(i,j+1)} \tilde{\Omega}_{(i+1,j)}^{\dagger}\ket{\tilde{\psi}_{\max}}=\ket{\tilde{\psi}_{\max}},
\end{aligned}
\end{equation}
which after multiplying by $\left(X_{(i,j+1)}^{2} X_{(i+1,j)}^{2} X_{(i+1,j+2)} X_{(i+2,j+1)}\right)^{\dagger}$ from the left side yields
\begin{equation}
\begin{aligned} 
&\Omega_{(i-1,j)} \Omega_{(i,j-1)}^{\dagger} \tilde{\Omega}_{(i,j+1)} \tilde{\Omega}_{(i+1,j)}^{\dagger}\ket{\tilde{\psi}_{\max}}=\left(X_{(i,j+1)}^{2} X_{(i+1,j)}^{2} X_{(i+1,j+2)} X_{(i+2,j+1)}\right)^{\dagger}\ket{\tilde{\psi}_{\max}}=\ket{\tilde{\psi}_{\max}}.
\end{aligned}
\end{equation}
Then we can compare the left-hand side of the above and Eq. \eqref{eq:stabilizing_one_omega}
\begin{equation}
\begin{aligned} 
&\Omega_{(i-1,j)} \Omega_{(i,j-1)}^{\dagger} \tilde{\Omega}_{(i,j+1)} \tilde{\Omega}_{(i+1,j)}^{\dagger}\ket{\tilde{\psi}_{\max}}=\Omega_{(i-1,j)} \Omega_{(i,j-1)}^{\dagger} \Omega_{(i,j+1)} \Omega_{(i+1,j)}^{\dagger}\ket{\tilde{\psi}_{\max}}.
\end{aligned}
\end{equation}
Finally, after multiplying the above by $\Omega_{(i-1,j)}^{\dagger} \Omega_{(i,j-1)}$ and expanding $\Omega_{(i,j+1)} \Omega_{(i+1,j)}^{\dagger},\tilde{\Omega}_{(i,j+1)}, \tilde{\Omega}_{(i+1,j)}^{\dagger}$ we arrive at
\begin{equation}
\begin{aligned} 
&\left(\omega^{2} Z_{(i,j+1)} \Pi^{1}_{(i,j+1)}+\omega Z_{(i,j+1)}^{2} \Pi^{2}_{(i,j+1)}\right)\left(\omega^{2} Z_{(i+1,j)} \Pi^{1}_{(i+1,j)}+\omega Z_{(i+1,j)}^{2}\Pi^{2}_{(i+1,j)}\right)^{\dagger}\ket{\tilde{\psi}_{\max}}\\
&=\left( Z_{(i,j+1)}\Pi^{1}_{(i,j+1)}+Z_{(i,j+1)}^{2}\Pi^{2}_{(i,j+1)}\right)\left( Z_{(i+1,j)}\Pi^{1}_{(i+1,j)}+ Z_{(i+1,j)}^{2}\Pi^{2}_{(i+1,j)}\right)^{\dagger}\ket{\tilde{\psi}_{\max}}.
\end{aligned}
\end{equation}
After a bit of rewriting, this simplifies to
\begin{equation}
\begin{aligned} 
\left((1-\omega) Z_{(i,j+1)}Z_{(i+1,j)}\Pi^{1}_{(i,j+1)}\Pi^{2}_{(i+1,j)}+(1-\omega^{2}) Z_{(i,j+1)}^{2}Z_{(i+1,j)}^{2}\Pi^{2}_{(i,j+1)}\Pi^{1}_{(i+1,j)}\right)\ket{\tilde{\psi}_{\max}}=0.
\end{aligned}
\end{equation}
Since $ \Pi^{1}_{(i,j+1)} \Pi^{2}_{(i+1,j)}$ and $ \Pi^{2}_{(i,j+1)}\Pi^{1}_{(i+1,j)}$ are projectors onto orthogonal spaces, for the equation above to be satisfied, both operators in the sum have to nullify $\ket{\tilde{\psi}_{\max}}$. In fact, an even stronger condition holds - since $Z_{(i,j+1)}Z_{(i+1,j)}$ is full rank, the projectors themselves nullify the state
\begin{equation}\label{eq:nullification_1}
\Pi^{1}_{(i,j+1)} \Pi^{2}_{(i+1,j)}\ket{\tilde{\psi}_{\max}}=0,\quad\Pi^{2}_{(i,j+1)}\Pi^{1}_{(i+1,j)}\ket{\tilde{\psi}_{\max}}=0.
\end{equation}
Moreover, if by using the operator $X_{(i-2,j+1)}^{2} X_{(i-1,j)}^{2} X_{(i-1,j+2)} X_{(i,j+1)}$ in \eqref{eq:projector_derivation}, we can derive the following relations
\begin{equation}\label{eq:nullification_2}
  \Pi^{1}_{(i,j+1)}\Pi^{2}_{(i-1,j)} \ket{\tilde{\psi}_{\max}}=0,\quad \Pi^{2}_{(i,j+1)}\Pi^{1}_{(i-1,j)} \ket{\tilde{\psi}_{\max}}=0,
\end{equation}
for all $(i,j)\in\plaquetteSet$. These four relations have a very interesting consequence, namely 
\begin{equation}\label{eq:projector_elimination}
\prod_{(k,l)\in\edgeSet}\left(\Pi^{1}_{(k,l)}+\Pi^{2}_{(k,l)}\right)\ket{\tilde{\psi}_{\max}}=\left(\prod_{(k,l)\in\edgeSet}\Pi^{1}_{(k,l)}+\prod_{(k,l)\in\edgeSet}\Pi^{2}_{(k,l)}\right)\ket{\tilde{\psi}_{\max}}.
\end{equation}
To show that this is the case, let us focus on an arbitrary operator $\Pi_{(i,j)}^{m}$ and consider a product
\begin{equation}
\prod_{(k,l)\in \edgeSet_{(i,j)}} \left(\Pi_{(k,l)}^{1}+\Pi_{(k,l)}^{2}\right),
\end{equation}
where $\edgeSet_{(i,j)}=\{(i,j),(i,j-2),(i,j+2),(i-2,j),(i+2,j)\}$ is the set consisting of $(i,j)$ and it's four neighbours. Then, by the virtue of Eq. \eqref{eq:nullification_1} and \eqref{eq:nullification_2}, the only term consisting of $\Pi_{(i,j)}^{k}$ that does not nulify the state is made up of other $k$ projectors
\begin{equation}
\prod_{(k,l)\in \edgeSet_{(i,j)}} \left(\Pi_{(k,l)}^{1}+\Pi_{(k,l)}^{2}\right)= \left(\prod_{(k,l)\in \edgeSet_{(i,j)}^{1}} \Pi_{(k,l)}^{1}+\prod_{(k,l)\in \edgeSet_{(i,j)}^{1}} \Pi_{(k,l)}^{2}\right).
\end{equation}
Next, we can extend that logic to the neighbours of neighbours of $(i,j)$ and so on. By repeating this procedure enough times such that we cover the whole set $\edgeSet$, we prove Eq. \eqref{eq:projector_elimination}.

Let us now come bact to the stabilizing relations \eqref{eq:stabilizing_with_omega} - since for all $(k,l)\in\edgeSet$ $Q_{1}^{(k,l)}+Q_{2}^{(k,l)}=\mathbb{I}$, we can rewrite the stabilizing conditions as
\begin{equation}
\begin{aligned} 
 X_{(i-1,j)}^{2} X_{(i,j-1)}^{2} X_{(i,j+1)} X_{(i+1,j)}\prod_{(k,l)\in\edgeSet}\left(\Pi^{1}_{(k,l)}+\Pi^{2}_{(k,l)}\right) \ket{\tilde{\psi}_{\max}}&= \ket{\tilde{\psi}_{\max}},\\
\Omega_{(i-1,j)} \Omega_{(i,j-1)}^{\dagger} \Omega_{(i,j+1)} \Omega_{(i+1,j)}^{\dagger}\prod_{(k,l)\in\edgeSet}\left(\Pi^{1}_{(k,l)}+\Pi^{2}_{(k,l)}\right)\ket{\tilde{\psi}_{\max}}&= \ket{\tilde{\psi}_{\max}}.
\end{aligned}
\end{equation}
Then, using Eq. \eqref{eq:projector_elimination} we can express the stabilizing relations as
\begin{equation}\label{eq:stabilizing_projectors}
\begin{aligned} 
 X_{(i-1,j)}^{2} X_{(i,j-1)}^{2} X_{(i,j+1)} X_{(i+1,j)}\left(\prod_{(k,l)\in\edgeSet}\Pi^{1}_{(k,l)} +\prod_{(k,l)\in\edgeSet}\Pi^{2}_{(k,l)}\right)\ket{\tilde{\psi}_{\max}}&= \ket{\tilde{\psi}_{\max}},\\
\left(Z_{(i-1,j)} Z_{(i,j-1)}^{2} Z_{(i,j+1)} Z_{(i+1,j)}^{2}\prod_{(k,l)\in\edgeSet}\Pi^{1}_{(k,l)} + Z_{(i-1,j)}^{2} Z_{(i,j-1)} Z_{(i,j+1)}^{2} Z_{(i+1,j)}\prod_{(k,l)\in\edgeSet}\Pi^{2}_{(k,l)}\right)\ket{\tilde{\psi}_{\max}}&= \ket{\tilde{\psi}_{\max}}.
\end{aligned}
\end{equation}

\textit{Part IV.}
To self-test the state $\ket{\tilde{\psi}_{\max}}$ let us consider its Schmidt decomposition 
\begin{equation}
\ket{\tilde{\psi}_{\max}} = \sum_{k=1}^{D} \alpha_{k}\ket{\phi_{k}}\ket{\eta_{k}},
\end{equation}
where $\alpha_{k}\geqslant 0$, $\ket{\phi_{k}}\in\mathcal{H}'$, $\ket{\eta_{k}}\in \mathcal{H}''$, and $D=\min (\dim(\mathcal{H}'),\dim(\mathcal{H}''))$. Let $V_{m}\subset \mathcal{H''}$ be a subspace associated to the projector $\prod_{(k,l)\in\edgeSet}\Pi^{m}_{(k,l)}$. Then, from from the stabilizing relations \eqref{eq:stabilizing_projectors} it follows that for each $k$ the state $\ket{\eta_{k}}$ is either in $V_{1}$ or $V_{2}$. To reflect that let us define $\ket{\eta_{m,k}}\in V_{m}$ to be a state $\ket{\eta_{k}}$ from $V_{m}$. This allows us to rewrite the Schmidt decomposition as 
\begin{equation}
\ket{\tilde{\psi}_{\max}} = \sum_{m=1}^{2}\sum_{k=1}^{D_{m}}\alpha_{m,k}\ket{\phi_{m,k}}\ket{\eta_{m,k}},
\end{equation}
where $D_{1}+D_{2}=D$, and $\alpha_{m,k}$ and $\ket{\phi_{m,k}}$ have an extra index $m$ that signifies to which $V_{m}$ the corresponding state $\ket{\eta_{m,k}}$ belongs. 

Substituting the above to Eq. \eqref{eq:stabilizing_projectors} yileds
\begin{equation}
\begin{aligned}
X_{(i-1,j)}^{2} X_{(i,j-1)}^{2} X_{(i,j+1)} X_{(i+1,j)} \sum_{k=1}^{D_{1}}\alpha_{1,k} \ket{\phi_{1,k}}\ket{\eta_{1,k}}&=\sum_{k=1}^{D_{1}}\alpha_{1,k} \ket{\phi_{1,k}}\ket{\eta_{1,k}},\\
Z_{(i-1,j)} Z_{(i,j-1)}^{2} Z_{(i,j+1)} Z_{(i+1,j)}^{2}\sum_{k=1}^{D_{1}}\alpha_{1,k} \ket{\phi_{1,k}}\ket{\eta_{1,k}}&=\sum_{k=1}^{D_{1}}\alpha_{1,k} \ket{\phi_{1,k}}\ket{\eta_{1,k}},\\
X_{(i-1,j)}^{2} X_{(i,j-1)}^{2} X_{(i,j+1)} X_{(i+1,j)} \sum_{k=1}^{D_{2}}\alpha_{2,k} \ket{\phi_{2,k}}\ket{\eta_{2,k}}&=\sum_{k=1}^{D_{2}}\alpha_{2,k} \ket{\phi_{2,k}}\ket{\eta_{2,k}},\\
Z_{(i-1,j)}^{2} Z_{(i,j-1)} Z_{(i,j+1)}^{2} Z_{(i+1,j)}\sum_{k=1}^{D_{2}}\alpha_{2,k} \ket{\phi_{2,k}}\ket{\eta_{2,k}}&=\sum_{k=1}^{D_{2}}\alpha_{2,k} \ket{\phi_{2,k}}\ket{\eta_{2,k}},
\end{aligned}
\end{equation}
where we used the fact that $\prod_{(k',l')\in\edgeSet}\Pi_{(k',l')}^{m}\ket{\eta_{m,k}}$ for all $k$ and $m$. Projecting the above relations onto the states $\ket{\eta_{m,k}}$ we show that 
\begin{equation}
\begin{aligned}
X_{(i-1,j)}^{2} X_{(i,j-1)}^{2} X_{(i,j+1)} X_{(i+1,j)} \ket{\phi_{1,k}}&=\ket{\phi_{1,k}},\qquad Z_{(i-1,j)} Z_{(i,j-1)}^{2} Z_{(i,j+1)} Z_{(i+1,j)}^{2}\ket{\phi_{1,k}}= \ket{\phi_{1,k}},\\
X_{(i-1,j)}^{2} X_{(i,j-1)}^{2} X_{(i,j+1)} X_{(i+1,j)} \ket{\phi_{2,k}}&=\ket{\phi_{2,k}},\qquad Z_{(i-1,j)} Z_{(i,j-1)}^{2} Z_{(i,j+1)} Z_{(i+1,j)}^{2}\ket{\phi_{2,k}}= \ket{\phi_{2,k}},\\
\end{aligned}
\end{equation}
for all $k$. However, this implies that $\ket{\phi_{1,k}}$ are states from a toric code and $\ket{\phi_{2,k}}$ are states from a complex conjugation of a toric code. Then, denoting the toric code basis as $\{\ket{\tau_{k}}\}_{k=0}^{3}$ we get
\begin{equation}
\bigotimes_{(i,j)\in\edgeSet}U_{(i,j)} \ket{\psi_{\operatorname{max}}} =  \sum_{k=1}^{4} \alpha_{1,k}\ket{\tau_{k}}\ket{\eta_{1,k}}+\alpha_{2,k}\ket{\tau_{k}^{*}}\ket{\eta_{2,k}},
\end{equation}
which ends the proof.
\end{proof}

\end{widetext}

\bibliography{toric-bell}%

\end{document}